\documentclass[%
 reprint,
superscriptaddress,
nofootinbib,
 amsmath,amssymb,
 aps,
pra,
]{revtex4-2}

\usepackage{graphicx}
\usepackage{dcolumn}
\usepackage{bm}
\usepackage{hyperref}

\usepackage{amsmath}
\usepackage{amssymb}
\usepackage{stackrel}
\usepackage{tikz-cd}
\usepackage{tikz}
\usepackage[cm]{fullpage}
\usepackage{amsthm}
\usepackage{mathtools}
\usepackage{enumitem}
\usepackage{bbold}
\usepackage[normalem]{ulem}
\usepackage{hyperref}
\usepackage{cleveref}
\usepackage{appendix}
\usepackage{comment}
\usepackage{subcaption}
\usepackage{centernot}
\usepackage{braket}
\usepackage{algorithm}
\usepackage{algpseudocode}
\usepackage{varwidth}
\usepackage{ulem}
\usepackage{apptools}
\AtAppendix{\counterwithin{Lemma}{section}}
\AtAppendix{\counterwithin{Theorem}{section}}
\AtAppendix{\counterwithin{Proposition}{section}}

\newtheorem{Theorem}{Theorem}
\newtheorem{Lemma}{Lemma}
\newtheorem{Proposition}{Proposition}
\newtheorem{Corollary}{Corollary}

\theoremstyle{definition}
\newtheorem{Definition}{Definition}

\DeclareMathOperator{\Odd}{Odd}

\DeclareMathOperator{\wt}{wt}
\DeclareMathOperator{\ad}{ad}

\begin{document}

\preprint{APS/123-QED}

\title{Minimally Universal Parity Quantum Computing}

\author{Isaac D. Smith}
\email{isaac.smith@uibk.ac.at}
\affiliation{University of Innsbruck, Department of Theoretical Physics, Technikerstr. 21A, Innsbruck A-6020, Austria}

\author{Berend Klaver}
\affiliation{University of Innsbruck, Department of Theoretical Physics, Technikerstr. 21A, Innsbruck A-6020, Austria}
\affiliation{Parity Quantum Computing GmbH, A-6020 Innsbruck, Austria}

\author{Hendrik Poulsen Nautrup}
\affiliation{University of Innsbruck, Department of Theoretical Physics, Technikerstr. 21A, Innsbruck A-6020, Austria}

\author{Wolfgang Lechner}
\affiliation{University of Innsbruck, Department of Theoretical Physics, Technikerstr. 21A, Innsbruck A-6020, Austria}
\affiliation{Parity Quantum Computing GmbH, A-6020 Innsbruck, Austria}
\affiliation{Parity Quantum Computing Germany GmbH, 20095 Hamburg, Germany}

\author{Hans J. Briegel}
\affiliation{University of Innsbruck, Department of Theoretical Physics, Technikerstr. 21A, Innsbruck A-6020, Austria}

\date{\today}

\begin{abstract} In parity quantum computing, multi-qubit logical gates are implemented by single-qubit rotations on a suitably encoded state involving auxiliary qubits. Consequently, there is a correspondence between qubit count and the size of the native gate set. One might then wonder: what is the smallest number of auxiliary qubits that still allows for universal parity computing? Here, we demonstrate that the answer is one, if the number of logical qubits is even, and two otherwise. Furthermore, we present a sufficient condition for a given parity gate set to be universal. This leads to a variety of different universal parity gate sets corresponding to different numbers of auxiliary qubits, and more generally contributes to the understanding of which entangling gates are required to augment the set of single-qubit unitaries to perform universal quantum computing. As a consequence, we obtain (i) minimal implementations of the parity framework on e.g., a triangular lattice, (ii) hardware specific implementations of the parity flow framework on e.g., a heavy-hex lattice, and (iii) novel universal resources for measurement-based quantum computation (MBQC).
\end{abstract}

\maketitle


\section{Introduction} \label{sec:intro}

In the near-term era of quantum computation, the number of physical qubits is low and the fidelity of, in particular, multi-qubit gates is relatively poor. Consequently, much work has gone into developing different gate sets and new techniques in order to reduce the required resources for specific quantum circuits and architectures.

The novelty of the parity quantum computing framework \cite{lechner_15,fellner_22} consists in the trade-off between physical qubit number and the ease of performing multi-qubit rotations. Specifically, a logical state on $n$ qubits is encoded to a new state using $k$ additional qubits, which allows multi-qubit logical rotations to be implemented via single-qubit rotations on the latter. In effect, the parity encoding assigns to each auxiliary qubit certain information pertaining to some subset of the logical qubits, which define the support of the logical rotation. This framework originates in the quantum annealing community with each $Z \otimes Z$ term of an Ising Hamiltonian being mapped to a separate additional qubit, and has since found application in quantum optimization \cite{Lanthaler_23,Ender2023parityquantum,DriebSchon2023parityquantum,Fellner2023parityquantum}.

As each auxiliary qubit, often called a `parity' qubit, corresponds to a specific multi-qubit rotation, there is a direct correspondence between the number of qubits and the native gate set implementable in the framework. In Ref.~\cite{fellner_22}, a universal gate set was presented consisting of single-qubit Pauli $X$ and $Z$ rotations for each of the $n$ logical qubits as well as $Z \otimes Z$-rotations between each pair of logical qubits, which is to say, one parity qubit for every two-qubit rotation involved in the computation. Accordingly, this gate set can be described by a generating set containing $\frac{1}{2}n(n+3)$ Hamiltonians, all of which are single- and two-body Pauli strings.

Recently, it was demonstrated that the minimal possible generating set containing only Pauli strings and permits universal computation, contains just $2n+1$ elements \cite{opt_gen_arxiv}. Due to the correspondence between elements of the generating set and the number of physical qubits in the parity framework, it is prudent to ask: how much can we minimize the number of parity qubits while still ensuring universality? 

In this work, we answer this question via two main results. In the first, we present a sufficient condition for universality based on the properties of the set of subsets of logical qubits whose parities are mapped to parity qubits by the parity encoding. This condition pertains to both the size of each subset as well as their mutual intersections. As a consequence, we obtain parity encodings for various numbers of auxiliary qubits that all permit universal computation. Included in these encodings are the minimal cases involving a single parity qubit when the number of logical qubits $n$ is even, and two parity qubits, when $n$ is odd. The second main result demonstrates the \textit{impossibility} of performing parity quantum computing with less than two parity qubits in the latter case, indicating that, in the parity quantum computing framework, the lower bound of Ref.~\cite{opt_gen_arxiv} can only be obtained in the case of $n$ even. 

The are a number of implications of these results. Since any generating set that contains a subset which satisfies the conditions outlined above is also universal by default, our results provide a method for demonstrating universality for parity encodings tailored to different physical layouts. For example, we demonstrate below that the generating set sufficient conditions can lead to a variety of possible arrangements on a triangular lattice using different numbers of parity qubits. Furthermore, due to the recent development of the parity flow framework \cite{klaver2024} which effectively performs parity quantum computation without the use of auxiliary qubits, our results produce a range of options for performing quantum algorithms in a manner suited to specific hardware constraints but without the need for SWAP gates. In particular, we present a universal circuit Ansatz tailored to nearest-neighbor interactions on a heavy-hex layout typical to devices such as those currently provided by IBM Quantum \cite{IBMQ}. Finally, as demonstrated in Ref.~\cite{smith_24}, there is a connection between parity quantum computation and measurement-based quantum computation (MBQC) \cite{Raussendorf_01,briegel2009measurement,raussendorf2001computational,Raussendorf_03}, which allows us to leverage the generating set sufficient conditions in defining families of universal resource states \cite{Nest_06} for the latter computing framework.

The remainder of this manuscript is structured as follows. In \Cref{sec:background}, we present a brief introduction to the relevant information from quantum control theory relevant for understanding universal generating sets for quantum computation as well as the parity quantum computation framework. In \Cref{sec:min}, we provide our main results, namely \Cref{thm:n_geq_2_main_text} and \Cref{thm:n_odd_no_single_main_text}, with the former presenting the sufficient conditions used throughout the rest of this work. In \Cref{sec:implications}, we investigate some implications of these results for implementing the parity quantum computation and the parity flow frameworks on different two-dimensional layouts of physical qubits. We conclude in \Cref{sec:disc} with discussion of possible further implications of our results. The proofs of the main results are given in the appendices.

\section{Background} \label{sec:background}

In this work, we are concerned with universal quantum computing within the parity computing framework. Accordingly, we require an understanding of what constitutes universal quantum computing, as well as how this looks in the context of parity quantum computing. The concepts related to the former are drawn primarily from the field of quantum control theory; the reader is directed to e.g.,  \cite{d2007introduction,huang1983controllability,albertini2001notions,nielsen2010quantum} for further information.

\subsection{Universal Quantum Computing}

For our present purposes, a quantum computer is taken to be a closed, finite dimensional quantum system with state space described by a Hilbert space $\mathcal{H}$. A quantum computation then consists of evolving an initial state $\ket{\psi_{\text{in}}} \in \mathcal{H}$, which represents the input to the computation, to a final state $\ket{\psi_{\text{out}}} \in \mathcal{H}$, which represents the logical output. Quantum theory tells us that the evolution taking $\ket{\psi_{\text{in}}}$ to $\ket{\psi_{\text{out}}}$ can be described by a (special) unitary $U_{\text{comp}}$ such that
\begin{align}
\ket{\psi_{\text{out}}} = U_{\text{comp}}\ket{\psi_{\text{in}}}.
\end{align}
Moreover, we know from Schrödinger's equation that $U_{\text{comp}}$ can be specified via reference to a (traceless, time-independent) Hamiltonian $H$ and an evolution time $t$ by
\begin{align}
U_{\text{comp}} = e^{-iH_{\text{comp}}t}.
\end{align}
From the perspective of quantum control theory, the operator $H_{\text{comp}}$ is considered to describe how the system is controlled, leading to the evolution described by $U_{\text{comp}}$.

Typically, the overall evolution $U_{\text{comp}}$ is constructed from a sequence of component evolutions, that is,
\begin{align}
U_{\text{comp}} = e^{-iH_{j_{k}}t_{k}}\dots e^{-iH_{j_{2}}t_{2}}e^{-iH_{j_{1}}t_{1}} \label{eq:sequence}
\end{align}
where the $H_{j_{i}}$ are all (traceless, time-independent) Hermitian operators drawn from a fixed set $\{H_{j}\}_{j}$ and the $t_{i}$ are positive real values. The set $\{H_{j}\}_{j}$ represents all the different ways the system can be controlled and may be determined by e.g., specific experimental considerations.

For a given system to act as a \textit{universal} quantum computer, we need to be able to select a sequence of controls, i.e. sequence of operators from $\{H_{j}\}_{j}$ and times for which they are applied, to produce $U_{\text{comp}}$ for \textit{any} pair of states $\ket{\psi_{\text{in}}}$ and $\ket{\psi_{\text{out}}}$. That is, we must be able to solve \Cref{eq:sequence} for every $U_{\text{comp}} \in SU(N)$, where $SU(N)$ denotes the Lie group of $N \times N$ special unitary matrices ($N$ denotes the dimension of the Hilbert space $\mathcal{H}$). Equivalently, we must be able to solve
\begin{align}
e^{-iH_{\text{comp}}t} = e^{-iH_{j_{k}}t_{k}}...e^{-iH_{j_{2}}t_{2}}e^{-iH_{j_{1}}t_{1}} \label{eq:H_sequence}
\end{align}
for all $t \in \mathbb{R}_{\ge 0}$ and all $iH_{\text{comp}} \in \mathfrak{su}(N)$, where $\mathfrak{su}(N)$ denotes the Lie algebra of 
$N \times N$ traceless, skew-Hermitian matrices \cite{hall2013lie}.

The consideration of the Lie algebra $\mathfrak{su}(N)$ (and sub-algebras thereof) is common when treating problems within quantum control theory, and several tools from Lie algebra theory will be useful for the questions of universality in which we are interested. In brief, a \textit{Lie algebra} $\mathfrak{g}$ is a (real) subspace of the space of $m \times m$ complex matrices $M_{m}(\mathbb{C})$ equipped with a Lie bracket $[\cdot, \cdot]: \mathfrak{g} \times \mathfrak{g} \rightarrow \mathfrak{g}$ that satisfies certain properties (see e.g., \cite{hall2013lie}). In the cases pertinent to this work, the Lie bracket is given by the matrix commutator, that is, for any $A, B \in M_{m}(\mathbb{C})$,
\begin{align}
[A, B] & := AB - BA. 
\end{align}
Using the Lie bracket, it is possible to a define a so-called adjoint map associated to each element $A$ of the Lie algebra, defined as
\begin{equation} \label{eq:ad_map_Lie}
\begin{gathered}
\ad_{A}: \mathfrak{g} \rightarrow \mathfrak{g} \\
B \mapsto [A,B].
\end{gathered}
\end{equation}
Below, it will also be useful to denote the $r$-fold composition of the adjoint map $\ad_{A}(\cdot)$ by $\ad_{A}^{(r)}(\cdot)$, i.e.,
\begin{align}
\ad_{A}^{(r)}(B) := \underbrace{[A, [A, [A, ...[A}_{r},B]]]].
\end{align}
Of particular importance for treating expressions such as that appearing on the right-hand side of \Cref{eq:H_sequence}, is the following formula, which is defined for any (matrix) Lie algebra $\mathfrak{g}$. For any $A,B \in \mathfrak{g}$, the formula states
\begin{align}
e^{A}e^{B}e^{-A} = e^{B + \sum_{r = 1}^{\infty} \frac{1}{r!}\ad_{A}^{(r)}(B)}. \label{eq:BCH_ish}
\end{align}
To see how this relates to \Cref{eq:H_sequence}, let us suppose for a moment that $U_{\text{comp}}$ can be produced by just two component evolutions, given by applying $H_{j_{1}}$ for a time of $t_{1}$ followed by applying $H_{j_{2}}$ for a time of $t_{2}$, meaning that
\begin{align}
e^{-iH_{\text{comp}}t} = e^{-iH_{j_{2}}t_{2}}e^{-iH_{j_{1}}t_{1}}e^{iH_{j_{2}}t_{2}}.
\end{align}
By taking $A = -iH_{j_{2}}t_{2}$ and $B = -iH_{j_{1}}t_{1}$ in \Cref{eq:BCH_ish} allows us to write
\begin{align}
-iH_{\text{comp}}t = - iH_{j_{1}}t_{1} + \sum_{r = 1}^{\infty} \frac{1}{r!} \ad_{-iH_{j_{2}}t_{2}}^{(r)}(-iH_{j_{1}}t_{1}).
\end{align}
In other words, we are able to write $iH_{\text{comp}}$ as a real linear combination of $iH_{j_{1}}t_{1}$, $iH_{j_{2}}t_{2}$ and sequences of nested commutators between them. For the general case where more than just two controls are used, the situation is analogous, where now the sequences of nested commutators may include more than two distinct elements.

We are thus able to state what we mean by universal quantum computation: the set of controls $\{H_{j}\}_{j}$  is \textit{universal} if linear combinations of elements of $\mathcal{G} := \{iH_{j} \}_{j}$ as well as nested commutators of elements of $\mathcal{G}$ generate all of $\mathfrak{su}(N)$. Defining 
\begin{align}
\mathcal{G}^{\ad^{(r)}} := \{ \ad_{G_{1}}...\ad_{G_{r}}(G_{r+1}) : G_{1},...,G_{r+1} \in \mathcal{G} \}
\end{align}
we can state this more formally as: $\{H_{j}\}_{j}$  is \textit{universal} if
\begin{align}
\textrm{span}_{\mathbb{R}}\left\{ \mathcal{G} \bigcup_{r = 1}^{\infty} \mathcal{G}^{\ad^{(r)}} \right\} = \mathfrak{su}(N).
\end{align}
Below, we will largely work with the set $\mathcal{G}$ rather than $\{H_{j}\}_{j}$ (i.e. with the skew-Hermitian operators rather than the Hermitian ones). The elements of $\mathcal{G}$ are called \textit{generators} and $\mathcal{G}$ itself will be called a \textit{generating set}.

\begin{figure*}[!t] 
\centering
\begin{subfigure}[b]{0.48\textwidth} 
\centering
\includegraphics[width=\textwidth]{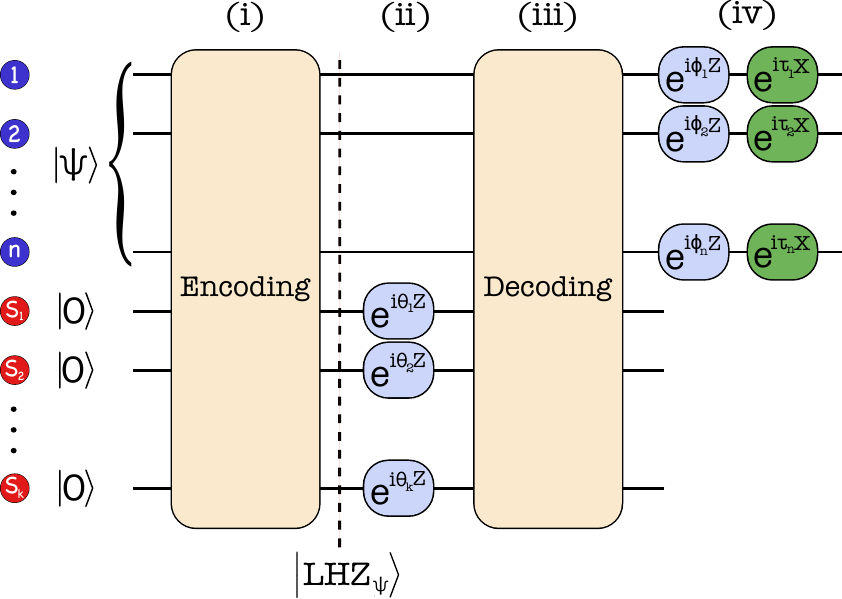} 
\caption{}
\label{fig:parity_phases}
\end{subfigure}
\hfill
\begin{subfigure}[b]{0.48\textwidth} 
\centering
\includegraphics[width=\textwidth]{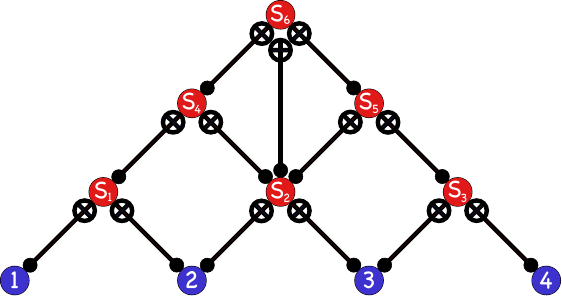} 
\vspace{1cm}
\caption{}
\label{fig:parity_original}
\end{subfigure} 
\caption{Parity quantum computing consists of four phases, depicted in figure (a), which are repeated for the duration of the computation. Initially, the state of the $n$ base qubits (in blue) represents the current logical state of the computation $\ket{\psi}$. The $k$ parity qubits (in red) are each initially prepared in the state $\ket{0}$ and are disentangled from the logical qubits. In phase (i), an encoding procedure is applied to the $n+k$ qubits to produce the encoded state denoted $\ket{\textrm{LHZ}_{\psi}}$. In phase (ii), single-qubit $Z$-rotations are applied to each of the $k$ parity qubits, and in phase (iii) a decoding procedure is performed through which the base qubits are in a new logical state and the parity qubits have been returned to the state $\ket{0}$. In the final phase, further single-qubit rotations are applied, completing the universal gate set. Figure (b) depicts the base and parity qubits for $n=4$ and $k=6$ arranged on a square lattice along with a unitary encoding/decoding procedure consisting of a sequence of CNOT gates between nearest-neighbors. For the encoding procedure, the CNOTs are applied in order from bottom to top. The parity sets are: $S_{1} = \{1,2\}$, $S_{2} = \{2,3\}$, $S_{3} = \{3,4\}$, $S_{4} = \{1,3\}$, $S_{5} = \{2,4\}$, and $S_{6} = \{1,4\}$.}
\label{fig:parity_QC}
\end{figure*}

\subsection{Parity Quantum Computing} \label{sec:parity}

Let us turn to the specific quantum computing framework, that of parity quantum computing \cite{lechner_15,fellner_22,Fellner_22_PRA}, the universality of which we are interested in investigating. In particular, we will see that the parity framework comes with a native family of generating sets, which will be the objects under consideration below.

Parity quantum computing typically proceeds by iteratively applying four phases: (i) an encoding phase where the current logical state is embedded in a larger Hilbert space, (ii) the application of physical single-qubit rotations on the extended space that result in logical multi-qubit rotations, (iii) a decoding procedure to obtain the new logical state on the original space, and (iv) the application of physical single-qubit rotations that result in logical single-qubit rotations. These four phases are depicted in \Cref{fig:parity_phases}.\footnote{It should be noted that there are several proposals for parity quantum computing in which a full decoding is not required in phase (iii), with the logical single-qubit Pauli-$X$ rotations implemented in a non-local fashion. As these proposals still use the same logical gate set, the specific details of these different proposals are not relevant here.} The generating set associated to this framework, called the \textit{parity generating set} and denoted by $\mathcal{G}_{\text{parity}}$ below, feature in phases (ii) and (iv); accordingly, we focus on those phases in the following. More information on the parity computing framework can be found in e.g., Refs.~\cite{lechner_15,fellner_22,Fellner_22_PRA}.

Let $\ket{\psi}$ be an $n$-qubit state representing the current logical state of the computation. At the commencement of phase (i), $n$ physical qubits, called \textit{base qubits} throughout, are in the state $\ket{\psi}$ while $k$ additional physical qubits, called \textit{parity qubits}, are each prepared in the $\ket{0}$ computational basis state. An encoding unitary $U_{\text{enc}}$ is then applied to all $n + k$ qubits to produce the state
\begin{align}
\ket{\textrm{LHZ}_{\psi}} = U_{\text{enc}}\ket{0}^{\otimes k}\ket{\psi}.
\end{align}
The encoding unitary $U_{\text{enc}}$ consists of a product of CNOT gates between base and parity qubits, arranged in such a way that the $k$ parity qubits encode parity information of $\ket{\psi}$ (whence the name parity qubits). For example, for $n = 2$ and $k = 1$, the state
\begin{align}
\ket{\psi} = \sum_{i,j = 0}^{1}\alpha_{ij}\ket{ij} 
\end{align}
is mapped to the state
\begin{align}
\ket{\textrm{LHZ}_{\psi}} = \sum_{i,j=0}^{1} \alpha_{i,j}\ket{i\oplus j}\ket{ij}
\end{align}
by the unitary $U_{\text{enc}}$ consisting of the product of two CNOT gates each with one of the base qubits as control and the parity qubit as target (the notation `$\oplus$' denotes modulo $2$ addition). For larger values of $n$ and $k$, $U_{\text{enc}}$ can be defined similarly; see e.g., \Cref{fig:parity_original} for a depiction of the encoding for $n = 4$ and $k = 6$. We would like to emphasize that the Hamiltonian corresponding to the unitary $U_{\text{enc}}$ is \textit{not} included in the generating set considered below as it does not directly enact a logical operation, but rather facilitates the implementation of the logical multi-qubit rotations via single-qubit rotations on the encoded parity qubits.

For our purposes, a useful perspective of the encoded state $\ket{\textrm{LHZ}_{\psi}}$ arises from stabilizer theory. Let us identify each of the $n$ base qubits with a label $i \in \{1,...,n\}$ and each of the $k$ parity qubits with a label given by a (non-empty) set $S_{j} \subseteq \{1,...,n\}$, $j = 1, ..., k$. This labeling convention also extends to unitary operations performed on data and parity qubits; that is, $Z_{i}$ denotes the single-qubit Pauli-$Z$ rotation on base qubit $i$ while $Z_{S_{j}}$ denotes a single-qubit Pauli-$Z$ rotation on the parity qubit labeled by $S_{j}$.\footnote{Note that, in the literature, the notation of a unitary indexed by a set is defined to be the tensor product of many copies of that unitary, one for each element in the set. We do not employ this notation here.} The idea is that a parity qubit labeled by $S_{j}$ encodes parity information of the base qubits whose labels are contained in the set $S_{j}$. Let us define $\mathbb{P} := \{S_{j} : j = 1, ..., k \}$. Then, for any $\ket{\psi}$, the state $\ket{\textrm{LHZ}_{\psi}}$ satisfies the following equation for each $S_{j} \in \mathbb{P}$:
\begin{align}
Z_{S_{j}} \bigotimes_{i \in S_{j}} Z_{i} \ket{\textrm{LHZ}_{\psi}} = \ket{\textrm{LHZ}_{\psi}}.
\end{align}
where $Z$ denotes the Pauli-Z operator (we denote the Pauli-X and Pauli-Y operators similarly by $X$ and $Y$), the label $S_{j}$ indicates the relevant parity qubit, and the labels $i \in S_{j}$ indicate the relevant base qubits. As a consequence of this, we have that
\begin{align}
e^{-iZ_{S_{j}}t}\ket{\textrm{LHZ}_{\psi}} = e^{-i(\bigotimes_{i \in S_{j}} Z_{i})t}\ket{\textrm{LHZ}_{\psi}}.
\end{align}
This means that, in effect, single-qubit $Z$-rotations performed on the parity qubits in the encoded state enact a multi-qubit rotation on the logical state. These rotations are precisely the operations pertaining to phase (ii) of the computational process. Regarding the generating sets we ultimately aim to consider, let us define
\begin{align}
\mathcal{G}_{\mathbb{P}} := \left\{i \bigotimes_{q \in S_{j}} Z_{q} | S_{j} \in \mathbb{P} \right\}.
\end{align}

Phase (iii) of the computation consists of the decoding phase where the logical state is mapped back to the $n$ base qubits. There are couple of equivalent ways in which this decoding procedure can occur: for example, one can apply the decoding unitary $U_{\text{dec}} = U_{\text{enc}}^{\dagger}$ as done in the original proposal for universal parity quantum computing \cite{fellner_22}; otherwise decoding can proceed by measuring each parity qubit in the Pauli $X$-basis and performing suitable corrections for certain measurement outcomes as suggested in Ref.~\cite{messinger_23}. The latter procedure has close links to measurement-based quantum computation (MBQC) \cite{Raussendorf_01,briegel2009measurement,raussendorf2001computational,Raussendorf_03} as was shown in \cite{smith_24} (see \Cref{app:mbqc} for a brief introduction to MBQC and an overview of the correspondence between the two frameworks). As the decoding phase has no bearing on the generating sets for the computation, we need not specify a particular decoding method here.

In the final phase, phase (iv), single-qubit rotations are applied to each of the base qubits, which are now disentangled from the parity qubits after the decoding step in phase (iii). Explicitly, we may apply $e^{-iX_{j}t}$ and $e^{-iZ_{j}t'}$ for each $j \in \{1, ..., n\}$ and any $t, t' \in \mathbb{R}_{\ge 0}$. Accordingly, we define the generating set pertaining to this phase as
\begin{align}
\mathcal{G}_{\text{s.q.}} := \{iX_{j}, iZ_{j} | j = 1, ..., n\} \label{eq:gen_set_sq}
\end{align}
where the subscript `s.q.' stands for `single-qubit'. Thus, the entire parity generating set is
\begin{align}
\mathcal{G}_{\text{parity}} = \mathcal{G}_{\text{s.q.}} \cup \mathcal{G}_{\mathbb{P}}.
\end{align}
In what follows, we will always consider $\mathcal{G}_{\text{parity}}$ to contain $\mathcal{G}_{\text{s.q.}}$ for the number of base qubits $n$, and investigate the consequences for universality as the set $\mathcal{G}_{\mathbb{P}}$ (equivalently $\mathbb{P}$) varies. In the original demonstration of the universality of the parity computing framework \cite{fellner_22}, the set $\mathbb{P}$ contained all pairs of labels from $\{1,...,n\}$, which we denote for later reference by
\begin{align}
\mathbb{P}_{\text{pairs}} := \left\{ \{i,j\} | i,j = 1,...,n \text{ s.t. } i < j \right\}. \label{eq:P_pairs}
\end{align}

\section{Minimal Universal Parity Generating Sets} \label{sec:min}

Using the notation introduced above, we can now state the question we are interested in: for which $\mathbb{P}$ is $|\mathbb{P}|$ minimized while maintaining that
\begin{align}
\textrm{span}_{\mathbb{R}}\left\{\mathcal{G}_{\text{parity}} \bigcup_{r = 1}^{\infty} \mathcal{G}_{\text{parity}}^{\ad^{(r)}}  \right\} = \mathfrak{su}(2^{n})? \label{eq:G_parity_su}
\end{align}
There are a couple of features of the generating set $\mathcal{G}_{\text{parity}}$ that are pertinent to the above question. First, this generating set consists entirely of Pauli strings, i.e., contains only terms consisting of products of Pauli operators. That is, defining 
\begin{align}
\mathcal{P}_{n} := \big\{ P_{1} \otimes P_{2} \otimes ... \otimes P_{n} | P_{i} \in \{I,X,Y,Z\} \big\}
\end{align}
we have that $\mathcal{G}_{\text{parity}} \subset i\mathcal{P}_{n}$. We use the notation $i\mathcal{P}_{n}$ to denote the set obtained by multiplying every element of $\mathcal{P}_{n}$ by the imaginary unit $i$ (note also that $\mathcal{P}_{n}$ is the set of Pauli strings, not the Pauli group on $n$ qubits). The inclusion of $\mathcal{G}_{\text{parity}}$ in $i\mathcal{P}_{n}$ is significant in part due to the fact that elements of $i\mathcal{P}_{n}$ exhibit nice commutation relations: any two elements of $i\mathcal{P}$ either commute or anti-commute, with the latter operator being proportional to an element of $i\mathcal{P}$. Noting that $i\mathcal{P}_{n}^{*} := i\mathcal{P}_{n} \setminus \{iI^{\otimes n}\}$ forms a basis of the space of traceless, skew-Hermitian operators, it follows that, to show \Cref{eq:G_parity_su} holds, it suffices in this case to demonstrate that\footnote{To be fully precise, the following notion of equality should be considered element-wise equality of the two sets up to a real constant. That is, all the elements in the left-hand set will be of the form $i2^{r}P$ while those in the right-hand set are of the form $iP$. This real constant is unimportant in light of the real span taken in \Cref{eq:G_parity_su}, so we continue with the abuse of notation here.}
\begin{align}
\mathcal{G}_{\text{parity}} \bigcup_{r = 1}^{\infty} \mathcal{G}_{\text{parity}}^{\ad^{(r)}} = i\mathcal{P}_{n}^{*}.
\end{align}

A second noteworthy feature of $\mathcal{G}_{\text{parity}}$ pertains to the commutation relations present in the subset $\mathcal{G}_{\text{s.q.}}$. In particular, each element of $\mathcal{G}_{\text{s.q.}}$ anti-commutes with precisely one other element: $iX_{j}$ anti-commutes with $iZ_{j'}$ if and only if $j = j'$, while $iX_{j}$ and $iX_{j'}$ (respectively $iZ_{j}$ and $iZ_{j'}$) commute for all $j,j' \in \{1,...,n\}$. Since we ultimately consider sequences of nested commutations between elements in $\mathcal{G}_{\text{s.q.}}$ (as well as $\mathcal{G}_{\mathbb{P}}$), these relations play a role in establishing the results presented below.

To begin to understand the possible limits on the minimal choices of $\mathbb{P}$, let us consider what is already known regarding generating sets consisting of Pauli strings. In Ref.~\cite{opt_gen_arxiv}, it was demonstrated that any universal generating set $\mathcal{G} \subset i\mathcal{P}_{n}^{*}$ is such that $|\mathcal{G}| \geq 2n+1$ and moreover that generating sets exist that obtain this bound.\footnote{This conclusion can also be drawn from results presented in Ref.~\cite{aguilar2024classificationpauliliealgebras}} Since 
\begin{align}
|\mathcal{G}_{\text{parity}}| = |\mathcal{G}_{\text{s.q.}}| + |\mathbb{P}|
\end{align}
and $|\mathcal{G}_{\text{s.q.}}| = 2n$, there is, at least in principle, nothing preventing the possibility of choosing $\mathbb{P}$ such that $|\mathbb{P}| = 1$. The results presented in the remainder of this section demonstrate when this is and isn't possible.

Our first result provides a sufficient condition on $\mathbb{P}$ that ensures the universality of $\mathcal{G}_{\text{parity}}$:
\begin{Theorem} \label{thm:n_geq_2_main_text} Let $n \geq 2$ and $\mathcal{G}_{\text{s.q.}}$ be as above. For any $1 \leq k\leq n-1$ and sets $S_{1}, ..., S_{k}$ such that 
\begin{enumerate}
    \item $|S_{j}|$ is even for all $j \in \{ 1, ..., k\}$,
    \item $\bigcup_{i=1}^{k} S_{i} = \{1,...,n\}$,
    \item if $k \geq 2$, then 
    \begin{enumerate}
        \item $S_{i} \cap S_{j} = \emptyset$ for all $1 \leq j \leq k$ such that $j \neq i-1,i, i+1$, and
        \item $S_{i} \cap S_{i+1} = \{s_{i}\}$ for all $i \leq k-1$, with the $s_{i} \in \{1,...,n\}$ all distinct.
    \end{enumerate}
\end{enumerate} 
For $\mathbb{P} = \{S_{j}| j = 1, ...,k\}$ and $\mathcal{G}_{\text{parity}} = \mathcal{G}_{\text{s.q.}} \cup \mathcal{G}_{\mathbb{P}}$, we have that 
\begin{align}
\mathcal{G}_{\text{parity}} \bigcup_{r = 1}^{\infty} \mathcal{G}_{\text{parity}}^{\ad^{(r)}} = i\mathcal{P}_{n}^{*}.
\end{align}
\end{Theorem}
The proof is given in \Cref{app:main_thm}. In particular, it makes use of a mapping between $i\mathcal{P}_{n}^{*}$ and the symplectic space $\mathbb{F}_{2}^{2N}$ which is presented in \Cref{app:mapping}.

There are number of consequences of this result for the parity computing framework. First of all, for any even $n$, taking $k = 1$ and $S = \{1,...,n\}$ satisfies the conditions of the theorem. Similarly, for odd $n$, it is possible to take $k = 2$ and define $S_{1}$ and $S_{2}$ that satisfy the required conditions (for example, taking $S_{1} = \{1,...,n-1\}$ and $S_{2} = \{n-1, n\}$ suffices). Accordingly, we have the following corollary:
\begin{Corollary} \label{cor:minimal_1_2} For $n \geq 2$, universal parity quantum computing is possible with just:
\begin{itemize}
    \item a single parity qubit, if $n$ is even, and
    \item two parity qubits, if $n$ is odd.
\end{itemize}
\end{Corollary}
In the case of even $n$, the generating set $\mathcal{G}_{\text{parity}}$ has $2n+1$ elements, so by the results of Ref.~\cite{opt_gen_arxiv} it is provably minimal. For the odd $n$ case, \textit{a priori} there is still a chance that we may be able to find a single element of $i\mathcal{P}_{n}^{*}$ which extends $\mathcal{G}_{\text{s.q.}}$ to a universal generating set. The following result demonstrates that this is not the case:
\begin{Theorem} \label{thm:n_odd_no_single_main_text} Let $n \geq 2$ be odd and $\mathcal{G}_{\text{s.q.}}$ be as above. For any $iP \in i\mathcal{P}_{n}^{*}$, $\mathcal{G} := \mathcal{G}_{\text{s.q.}} \cup \{iP\}$ is such that
\begin{align}
\mathcal{G} \bigcup_{r=1}^{\infty} \mathcal{G}^{\ad^{(r)}} \subsetneq i\mathcal{P}_{n}^{*}.
\end{align}
\end{Theorem}

The proof is given in \Cref{app:thm_n_odd_no_single_main_text} and again makes use of the mapping mentioned above. This result demonstrates two things. First, it completes the claim made earlier that for $n$ even and odd respectively, the minimal number of parity qubits required for universal computation are $1$ and $2$. In the former case, the resultant generating set has size $2n+1$, which we know is optimal, whereas the requirement of at least $2n+2$ generators in the latter case demonstrates that not every generating set of $2n$ Pauli strings can be extended to a universal set by simply appending a single additional Pauli string.

There are two final observation worth making here. The first is that, for any $\mathbb{P}$ that contains a subset of sets that satisfies the conditions of \Cref{thm:n_geq_2_main_text}, the generating set $\mathcal{G}_{\text{s.q.}} \cup \mathcal{G}_{\mathbb{P}}$ is universal. For example, if we consider the set $\mathbb{P}_{\text{pairs}}$ that was used in the original proof of universality of the parity framework \cite{fellner_22}, we see that it contains the sets
\begin{align}
\big\{ \{i, i+1\} | i = 1, ..., n-1 \big\}
\end{align}
which satisfy the conditions of \Cref{thm:n_geq_2_main_text}. This provides an alternative proof of the universality of $\mathcal{G}_{\text{s.q.}} \cup \mathcal{G}_{\mathbb{P}_{\text{pairs}}}$ to that given in Ref.~\cite{fellner_22} and also demonstrates that, from the point of view of universality, there is redundancy in the set of parities $\mathbb{P}_{\text{pairs}}$. 

The second observation is that, even though there are sets $\mathbb{P}$ satisfying the conditions of \Cref{thm:n_geq_2_main_text} which correspond to generating sets with greater than $2n+1$ elements, they can still be considered to be minimal from a certain perspective. Specifically, they are minimal in the sense that by removing any single element of $\mathbb{P}$, the generating set then fails to be universal. Again this contrasts to the case of $\mathbb{P}_{\text{pairs}}$ from which it is possible to omit elements and still maintain universality.

\section{Implications} \label{sec:implications}

The parity sets that satisfy the sufficient conditions of \Cref{thm:n_geq_2_main_text} are distinct from those considered to date in the literature. In this section we elucidate certain implications of using such parity sets within the parity quantum computing framework. In particular, we first demonstrate that, for the generating sets $\mathcal{G}_{\textrm{parity}}$ where $|\mathbb{P}| = 1$ if $n$ is even and $|\mathbb{P}|=2$ if $n$ is odd, it is possible to implement the rotation of {\em any} $n$-qubit Pauli string in constant depth. Thereafter, we focus on the ability to perform parity quantum computing using nearest-neighbor interactions between physical qubits arranged on a triangular lattice and on the ability to natively implement the recently developed parity flow framework \cite{klaver2024} suited to the specific connectivity of current quantum devices, such as those provided by IBM Quantum \cite{IBMQ}. Finally, we briefly outline the implications for universal measurement-based quantum computing.

\subsection{Implementing Pauli String Rotations in Constant-depth} \label{subsec:const_dep}

One pertinent question to ask when considering the universal generating sets related to \Cref{thm:n_geq_2_main_text} is how well they scale in terms of compiling specific unitaries. In this subsection, we provide a partial answer to this question for the minimal possible universal sets, that is, for the minimal sizes of $\mathbb{P}$.

Consider the set $\mathcal{G}_{\text{parity}}$ with 
\begin{align}
\mathbb{P} = \begin{cases} \{i Z_{1} \otimes \dots \otimes Z_{n}\}, &n = 0 \bmod 2 \\
\{i Z_{1} \otimes \dots \otimes Z_{j}, iZ_{j} \otimes \dots \otimes Z_{n}\}, & n = 1 \bmod 2 
\end{cases} \label{eq:min_parity_cases}
\end{align} 
for some even number $j \in \{2, \dots, n-1\}$. As a consequence of the results presented in \Cref{sec:min}, we know that, for any even $n \in \mathbb{N}_{\ge 2}$, $\mathcal{G}_{\text{parity}}$ defined in this way is universal. Apart from the dependence on $n$ being even or odd, the {\em size} of $\mathbb{P}$ doesn't depend on $n$. This is in contrast to other choices of parity set such as $\{i Z_{j} \otimes Z_{j+1} | j =1,...,n-1\}$. This lack of independence of $|\mathbb{P}|$ on $n$ may seem innocuous, but it has an interesting consequence for implementing unitaries of the form $e^{i\theta P}$ for $P \in \mathcal{P}_{n}^{*}$. Specifically, {\em any} such rotation can be implemented in constant depth using the generating set $\mathcal{G}_{\text{parity}}$ with $\mathbb{P}$ as above.

This claim follows from a part of the proof of \Cref{thm:n_geq_2_main_text} (namely \Cref{lem:z_A_improved}), which we elaborate upon below. Before doing so, let us comment on what ``depth'' means in this context. Clearly, as $n$ increases, the size of $\mathcal{G}_{\text{s.q.}}$ does also. Accordingly, the value $R < \infty$ for which 
\begin{align}
\mathcal{G}_{\text{parity}} \bigcup_{r = 1}^{R} \mathcal{G}_{\text{parity}}^{\ad^{(r)}} = i\mathcal{P}_{n}^{*}
\end{align}
will increase as $n$ increases. The value $R$ corresponds to the minimum value such that there always exists a sequence $G_{1}, \dots, G_{R'} \in \mathcal{G}_{\text{parity}}$ with $R' \leq R$ such that
\begin{align}
\ad_{G_{1}} \dots \ad_{G_{R'-1}}(G_{R'}) = \frac{1}{2^{R'-1}}iP \label{eq:depth_sequences}
\end{align}
for each $iP \in i\mathcal{P}_{n}^{*}$. However, what this value $R$ {\em doesn't} take into account is that most of the elements $G_{r}$ in the sequence will come from $\mathcal{G}_{\text{s.q.}}$, that is, they are single-qubit operators and hence can be implemented in parallel in a circuit. In light of this, we can define the {\em sequence depth} of implementing $iP$ as the minimum number of multi-qubit Pauli-string rotations from the parity set $\mathcal{G}_{\mathbb{P}}$ required to generate $iP$ by nested commutation. Let us consider the corresponding quantity 
\begin{align}
\text{seq-depth}(iP) := \min_{\substack{G_{1},\dots, G_{R'} \\ \text{s.t. (\ref{eq:depth_sequences}) holds} }} |\{G_{r} \in \mathcal{G}_{\mathbb{P}}| r = 1, ..., R' \}|.
\end{align}
We have the following:
\begin{Proposition} \label{prop:const_depth} Let $\mathcal{G}_{\text{parity}}$ be as in \Cref{eq:min_parity_cases}. For any $iP \in i\mathcal{P}_{n}^{*}$,
\begin{align}
\textrm{seq-depth}(iP) \leq \begin{cases} 3, & n = 0 \bmod 2, \\
  6, & n = 1 \bmod 2.
\end{cases}
\end{align}
\end{Proposition}
The proof follows as a corollary of \Cref{lem:z_A_improved}. The proof of \Cref{lem:z_A_improved} is constructive in the sense that, for each $iP \in i\mathcal{P}_{n}^{*}$, it produces $G_{1}, ..., G_{R'} \in \mathcal{G}_{\text{parity}}$ that satisfy \Cref{eq:depth_sequences}. For the even $n$ case, these sequences are given explicitly below. 

Before giving the sequences, let us consider the question: how does one actually go from such a sequence to the circuit implementing $e^{i\theta P}$? One possible answer is given by \Cref{eq:BCH_ish} presented earlier. As demonstrated in, e.g., Ref.~\cite{opt_gen_arxiv}, if the sequence $G_{1}, \dots, G_{R'}$ satisfies \Cref{eq:depth_sequences} then 
\begin{align}
e^{\frac{\pi}{4}G_1}\cdots e^{\frac{\pi}{4}G_{R'-1}} e^{\theta G_{R'}}e^{-\frac{\pi}{4}G_{R'-1}}\cdots e^{-\frac{\pi}{4}G_{1}} = e^{i\theta P}. \label{eq:Clifford_BCHish}
\end{align}
Pursuant to \Cref{prop:const_depth}, even though this product contains $2R'-1$ terms, and $R'$ grows with $n$, at most $5$ (resp. $11$) of them are entangling unitaries for all even (resp. odd) $n$. Thus, the implementation of $e^{i\theta P}$ is constant also in circuit depth, if each entangling gate and each parallel implementation of single-qubit gates are considered to have constant depth.\footnote{There exist measurement-based implementations of all the entangling gates considered here that have constant depth (see, e.g., Ref.~\cite{bäumer2024measurementbasedlongrangeentanglinggates}).}

The sequences for a given $iP \in i\mathcal{P}_{n}^{*}$ have a different structure depending on whether $n$ is even or odd and on whether the Pauli string $P = P_{1} \otimes \dots \otimes P_{n}$ has an even or odd number of tensor factors. We focus on the even $n$ case here; the odd case can be obtained by applying \Cref{lem:z_A_improved} twice. Let us write $\bm{Z} := Z_{1} \otimes \dots \otimes Z_{n}$ for the single parity operator. Let us write $\{j_{1}, \dots, j_{l}\} \subseteq \{1,\dots,n\}$ for the subset of labels $j$ for which $P_{j} \neq I$. Similarly, we write $\{k_{1}, \dots, k_{n-l}\} := \{1,\dots,n\} \setminus \{j_{1}, \dots, j_{l}\}$. If $l$ is odd, then the sequence
\begin{align}
G'_{1},\dots, G'_{l+2} = i\bm{Z},iX_{j_{1}},\dots,iX_{j_{l}},i\bm{Z} \label{eq:n_even_minimal_supp_even}
\end{align}
is such that
\begin{align}
\ad_{G'_{1}}\dots\ad_{G'_{l+1}}(G'_{l+2}) \propto iP'_{1} \otimes \dots \otimes P'_{n}
\end{align}
where $P'_{j} = X_{j}Z_{j}$ for each $j \in \textrm{supp}(P)$ and $P'_{j} = I$ otherwise, that is, the Pauli string on the right-hand side has the same support as $P$. Accordingly, there exists a sequence $G''_{1},\dots, G''_{t}$ consisting only of elements of $\mathcal{G}_{\text{s.q.}}$ such that
\begin{align}
\ad_{G''_{1}}\dots\ad_{G''_{t}}(iP'_{1} \otimes \dots \otimes P'_{n}) \propto iP.
\end{align}
The sequence $G_{1},\dots,G_{r}$ is then taken to be $G''_{1},\dots,G''_{t},G'_{1},\dots,G'_{l+2}$.

In the case, where $l$ is even, the sequence
\begin{align}
G'_{1},\dots,G'_{3n-3l+5} &= iX_{j_{1}},i\bm{Z}, iX_{k_{1}},\dots,iX_{k_{n-l}},\nonumber\\
& \quad \,\,\, iZ_{k_{1}},\dots,iZ_{k_{n-l}}, i\bm{Z},iX_{j_{1}}, \nonumber \\
& \quad \,\,\,\, iX_{k_{1}},\dots,iX_{k_{n-l}},i\bm{Z} \label{eq:n_even_minimal_supp_odd}
\end{align}
is such that
\begin{align}
\ad_{G'_{1}}\dots\ad_{G'_{3n-3l+4}}(G'_{3n-3l+5}) \propto i\hat{P}_{1} \otimes \dots \otimes \hat{P}_{n}
\end{align}
where $\hat{P}_{j} = Z_{j}$ if $j \in \textrm{supp}(P)$ and $\hat{P}_{j} = I$ otherwise. Similarly to above, it means that $\hat{P} = \hat{P}_{1} \otimes \dots \otimes \hat{P}_{n}$ has the same support as $P$, so there again exist $G''_{1}, \dots, G''_{t} \in \mathcal{G}_{\text{s.q.}}$ such that
\begin{align}
\ad_{G''_{1}}\dots\ad_{G''_{t}}(i\hat{P}_{1} \otimes \dots \otimes \hat{P}_{n}) \propto iP.
\end{align}
The sequence $G_{1},\dots,G_{r}$ is again taken to be the concatenation of the two sequences.

If we substitute these sequences into \Cref{eq:Clifford_BCHish}, and abbreviate notation to highlight the splitting into local Clifford rotations and global rotations, we get, in the case where $P$ has even support,
\begin{align}
e^{i\theta P} =  U_{\text{l.c.}}\, e^{i\frac{\pi}{4} \bm{Z}} \,V_{\text{l.c.}} \, e^{i\theta \bm{Z}} \, V_{\text{l.c.}}^{\dagger} \, e^{-i\frac{\pi}{4} \bm{Z}} \, U_{\text{l.c.}}^{\dagger},
\end{align}
where $U_{\text{l.c.}}$ and $V_{\text{l.c.}}$ are products of local Clifford operations. In the case where $P$ has odd support, we get that $e^{i\theta P}$ is produced by the product of unitaries
\begin{align}
\hat{U}_{\text{l.c.}}\, e^{i\frac{\pi}{4} \bm{Z}}\, \hat{V}_{\text{l.c.}}\,e^{i\frac{\pi}{4} \bm{Z}}\,\hat{W}_{\text{l.c.}}\,e^{i\theta \bm{Z}} \,\hat{W}_{\text{l.c.}}^{\dagger}\, e^{-i\frac{\pi}{4} \bm{Z}}\,V_{\text{l.c.}}^{\dagger}\,e^{-i\frac{\pi}{4} \bm{Z}}\,U_{\text{l.c.}}^{\dagger}
\end{align}
where $\hat{U}_{\text{l.c}}$, $\hat{V}_{\text{l.c.}}$ and $\hat{W}_{\text{l.c}}$ are also all products of local Clifford operations. The analogous sequences for the case where $n$ is odd are longer, which is due to the fact that, in such cases, $\mathbb{P}$ contains two elements (this is the source the factor of two difference between the sequences depths in \Cref{prop:const_depth}).

\subsection{Arrangement on a Triangular Lattice}

\begin{figure*}[!t] 
\centering
\begin{subfigure}[b]{0.8\textwidth} 
\centering
\includegraphics[width=0.8\textwidth]{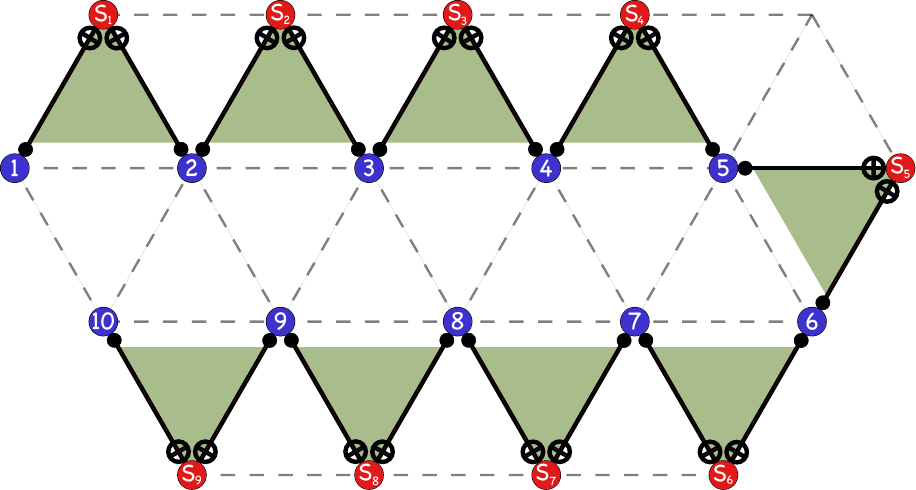}
\caption{}
\label{fig:2parity}
\end{subfigure} \\
\begin{subfigure}[b]{0.48\textwidth} 
\centering
\includegraphics[width=\textwidth]{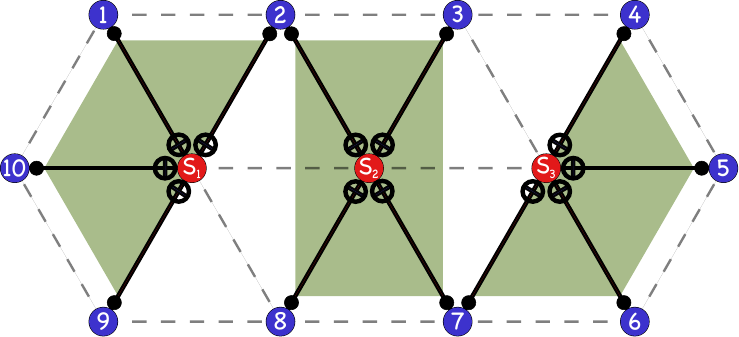} 
\caption{}
\label{fig:4parity}
\end{subfigure}
\hfill 
\begin{subfigure}[b]{0.48\textwidth} 
\centering
\includegraphics[width=\textwidth]{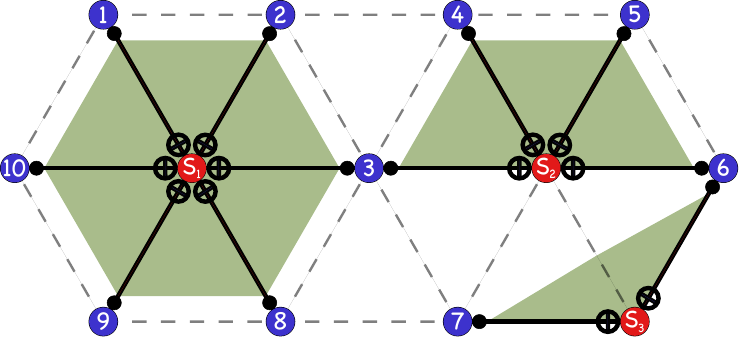} 
\caption{}
\label{fig:mixed_parity}
\end{subfigure} 
\caption{This figure depicts several possibilities for performing universal parity quantum computing using $10$ base qubits and varying numbers of parity qubits arranged on a triangular lattice. In each case, the base qubits are shown in blue, the parity qubits in red and the green shaded indicate which base qubits pertain to each parity set, i.e., each blue base qubit touching a corner of a green shape indicates that that qubit features in the parity set for that shape. For example, in figure (a) all the sets $S_{i}$ have two parities, namely $S_{i} = \{i, i+1\}$; in figure (b), each parity set has four elements with $S_{1} = \{1,2,9,10\}$, $S_{2} = \{2,3,7,8\}$ and $S_{3} = \{4,5,6,7\}$; in figure (c), the parity sets each have differing numbers of elements with $S_{1} = \{1,2,3,8,9,10\}$, $S_{2} = \{3,4,5,6\}$ and $S_{3} = \{6,7\}$. Each of these choices of sets satisfies the conditions of \Cref{thm:n_geq_2_main_text} and hence each layout supports universal parity quantum computing.}
\label{fig:10_data_qubit_figs}
\end{figure*}

One advantage of the parity framework is that long-range multi-qubit logical operations are performed by single-qubit physical rotations on an encoded state, which was moreover prepared using nearest-neighbor interactions. For example, when computing using the layout depicted in \Cref{fig:parity_original}, it is possible to implement the logical operation $R_{Z_{1} \otimes Z_{4}}(\theta)$ on the non-neighboring qubits $1$ and $4$, by performing the rotation $R_{Z_{\{1,4\}}}(\theta)$ on the parity qubit related to the set $\{1,4\}$. As the figure demonstrates, the encoded state permitting this can be prepared by arranging the physical base and parity qubits on a square lattice and then by applying a sequence of CNOT gates between nearest neighbors.

Here, we consider what conclusions can be drawn from \Cref{thm:n_geq_2_main_text} for designing possible layouts for the parity framework. In particular, we focus on arranging the physical base and parity qubits on a triangular lattice, while still requiring that the encoding unitary $U_{\text{enc}}$ be implemented via nearest-neighbor CNOT gates. In \Cref{fig:10_data_qubit_figs}, we portray several examples of possible layouts for $n=10$ base qubits corresponding to different choices of the sets $S_{1}, ..., S_{k}$ from \Cref{thm:n_geq_2_main_text}. 

There are a couple of benefits of performing parity computing using the layouts presented here. First, the layout in e.g.,  \Cref{fig:2parity} can be extended to any number of base qubits and in each case the corresponding unitary encoding has constant circuit depth. This contrasts with the circuit depth for the encoding unitary for the layouts of the form depicted in \Cref{fig:parity_original}, corresponding to the parity set $\mathbb{P}_{\text{pairs}}$, which scales with $n$ (although a constant-depth encoding using measurements does exist \cite{messinger_23}). 

Second, the layouts in \Cref{fig:10_data_qubit_figs} may be advantageous for physical device design in certain architectures. For example, in each of the layouts depicted, the topology ensures that it is possible to draw a path from the perimeter of the lattice to every two-qubit gate without crossing any other two-qubit gate. In the context of, e.g., superconducting qubits, this could allow for chip design where no reliance on so-called air-bridge crossovers \cite{chen2014fabrication,dunsworth2018method} for the required control lines for the implementation of two-qubit gates, thereby simplifying the fabrication process. Moreover, in such a context it is often common to implement two-qubit gates between a control qubit with high frequency and a target qubit with low frequency \cite{Sheldon_16,Paik_16}. Since the layouts in \Cref{fig:10_data_qubit_figs} all have fixed orientations of CNOTs (certain qubits are only ever targets and others are only ever controls), this suggests a natural distribution of high and low frequency qubits throughout the chip.

\subsection{Minimal Parity Flow}

\begin{figure*}[!t] 
\centering
\begin{subfigure}[b]{0.5\textwidth} 
\centering
\includegraphics[width=\textwidth]{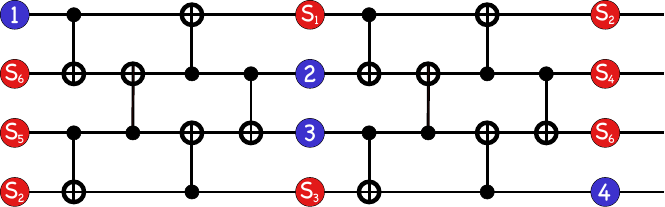} 
\caption{}
\label{fig:parity_flow_n4_k6}
\end{subfigure}
\hfill
\begin{subfigure}[b]{0.46\textwidth} 
\centering
\includegraphics[width=\textwidth]{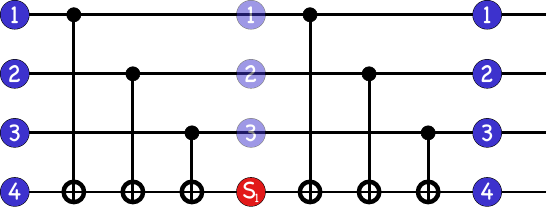} 
\caption{}
\label{fig:parity_flow_n4_min}
\end{subfigure} \\
\vspace{0.5cm}
\begin{subfigure}[b]{0.48\textwidth} 
\centering
\includegraphics[width=\textwidth]{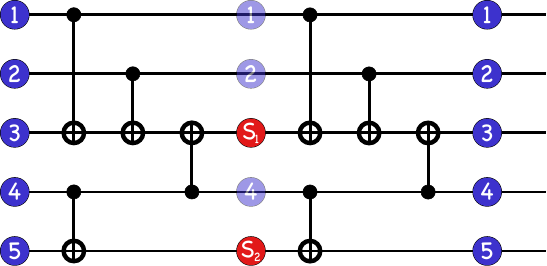} 
\caption{}
\label{fig:parity_flow_n5_min}
\end{subfigure}
\hfill
\begin{subfigure}[b]{0.48\textwidth} 
\centering
\includegraphics[width=\textwidth]{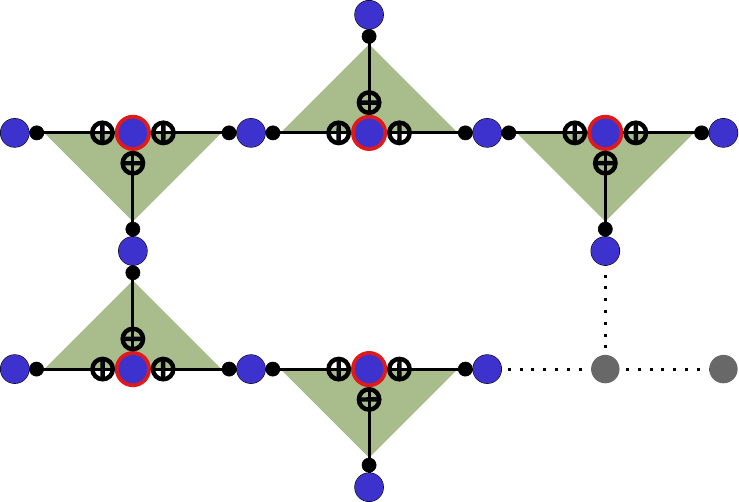} 
\caption{}
\label{fig:parity_flow_brickwork}
\end{subfigure} 
\caption{In the parity flow framework of Ref.~\cite{klaver2024}, the parity information of the $n$ base qubits is tracked through time in an $n$-qubit circuit. A circuit Ansatz permitting universal computation can be obtained from the figures (a) - (c) by inserting single-qubit $Z$-rotations at every location containing a red circle, and single-qubit $Z$-rotations and certain unitaries for implementing single-qubit $X$-rotations at every location containing a solid blue circle (in most cases, ``certain unitaries'' means the single-qubit $X$-rotation itself, but depending on the parity flow, extra CNOTs may be required - see \cite{klaver2024}). The faded blue circles indicate the parity information for the corresponding qubit at the given point in the circuit, but no rotations are required to be placed there for universality. In figure (a), the flow of parity information is depicted for the computation corresponding to the parity encoding shown in \Cref{fig:parity_original} for $n = 4$ and $k=6$. As can be seen, all four base qubits and all six parity sets $S_{1},...,S_{6}$ are present at least once in the circuit fragment shown (the parity sets are the same as in \Cref{fig:parity_original}). In figure (b), the flow of parities for the minimal number of parity sets for $n=4$ qubits is shown, i.e., with $S_{1} = \{1,2,3,4\}$. In figure (c), the flow of parities for the minimal number of parity sets for $n=5$ is shown, namely with $S_{1} = \{1,2,3,4 \}$ and $S_{2} = \{4,5\}$. In figure (d), the parity flow formalism is mapped to a heavy-hex layout, typical of the quantum devices provided by IBM Quantum \cite{IBMQ}. Each parity set contains four elements indicated by the four qubits touching or contained within each green triangle. The qubits colored both blue and red are those that represent both base and parity qubits at various times throughout the circuit, while all blue qubits only represent base qubits. The gray qubits and dotted lines indicate part of the heavy-hex lattice not involved in the computation.}
\label{fig:parity_flow_egs}
\end{figure*}

So far, we have presented the parity computing framework as encoding an $n$-qubit state as an $n+k$-qubit state. As developed in the parity flow formalism of Ref.~\cite{klaver2024}, it is in fact possible to perform parity quantum computing on $n$ qubits by appropriately tracking how the parity information changes ``in place'', that is, by allowing a given physical qubit to act as a parity qubit at various times throughout the circuit.

Let us consider an explanatory example. Recall from above that, for the case $n=2$ and $k=1$, the encoded state for the parity computation with logical state $\ket{\psi} = \sum_{i,j = 0}^{1} \alpha_{ij}\ket{i}\ket{j}$ is
\begin{align}
\ket{\textrm{LHZ}_{\psi}} = \sum_{i,j=0}^{1}\alpha_{ij}\ket{i\oplus j}\ket{i}\ket{j}. \label{eq:n2k1}
\end{align}
By applying $R_{Z}(\theta)$ to the parity qubit and then decoding (either unitarily or via measurements), we obtain the logical state $R_{Z \otimes Z}(\theta)\ket{\psi}$. Equivalently, we could avoid the use of the parity qubit altogether by considering instead the state 
\begin{align}
(\ket{0}\!\!\bra{0} \otimes I + \ket{1}\!\!\bra{1} \otimes X) \ket{\psi} = \sum_{i,j=0}^{1} \alpha_{ij}\ket{i}\ket{i\oplus j} \label{eq:in_place}
\end{align}
instead of $\ket{\textrm{LHZ}_{\psi}}$. In this case, we can arrive at the same final logical state by applying the rotation $I \otimes R_{Z}(\theta)$ followed by the CNOT gate $(\ket{0}\!\!\bra{0} \otimes I + \ket{1}\!\!\bra{1} \otimes X)$ \footnote{This example, as well as its extension of implementing a $Z \otimes Z \otimes .... \otimes Z$-rotation by a sequence of CNOTs and a single $Z$-rotation, has appeared in a number of places in the literature; an inexhaustive list includes the Refs.~ \cite{browne2006onewayquantumcomputation,Cowtan_2020,Ferguson_21,kaldenbach2023mappingquantumcircuitsshallowdepth,bäumer2024measurementbasedlongrangeentanglinggates}}. By inspecting \Cref{eq:n2k1} and \Cref{eq:in_place}, we see that the parity information that is mapped to the auxiliary qubit in the former case, is encoded in the second base qubit in the latter. In Ref.~\cite{klaver2024}, this line of reasoning is extended to the full parity encoding (i.e. for $\mathbb{P} = \mathbb{P}_{\text{pairs}}$ as in \Cref{fig:parity_original}), where the resource benefits for implementing certain algorithms were also shown. \Cref{fig:parity_flow_n4_k6} depicts an example of how the parity information changes for the parity computation using the parity encoding shown in \Cref{fig:parity_original} using $n = 4$ physical qubits.

We know from \Cref{thm:n_geq_2_main_text} above that smaller parity sets than $\mathbb{P}_{\text{pairs}}$ still allow for universal computation. In the formalism where parity sets are mapped to additional physical qubits, this resulted in a lower total qubit count. But what does it mean in the parity flow formalism where $n$ qubits are used in any case? In essence, different choices for $\mathbb{P}$ correspond to different choices of the number and orientation of the CNOTs comprising the unitaries that transport the parity information (i.e., the unitaries between the filled circles in the circuits depicted in \Cref{fig:parity_flow_egs}). For example, \Cref{fig:parity_flow_n4_min} and \Cref{fig:parity_flow_n5_min} depict the unitaries and corresponding flow of parity information for the minimal parity set examples for $n = 4$ and $n = 5$ qubits respectively. As can be seen, there is a reduced CNOT cost for implementing the parity flow for the minimal set $\mathbb{P}$ for $n=4$ base qubits as compared to the original case with $\mathbb{P}_{\text{pairs}}$. 

To obtain a circuit Ansatz allowing universal computation from each of the circuit structures presented in \Cref{fig:parity_flow_egs}, we can associate the different elements of the generating sets $\mathcal{G}_{\text{parity}}$ to the different components of the circuit. Just as with the original parity quantum computing framework, we can obtain a universal computation by inserting single-qubit rotations in the appropriate locations. For example, by placing single-qubit $Z$- and $X$-rotations at the location of every solid blue circle in \Cref{fig:parity_flow_n4_k6,fig:parity_flow_n4_min,fig:parity_flow_n5_min}, we cover the elements of $\mathcal{G}_{\text{parity}}$ corresponding to $\mathcal{G}_{\text{s.q.}}$, while placing a single-qubit $Z$-rotation at the location of every red circle covers the elements corresponding to $\mathcal{G}_{\mathbb{P}}$. Accordingly, the key differences between the different parity sets in the parity flow picture amount to differences in the two-qubit gate placement of the corresponding circuits.

Similarly to the previous subsection, it is possible to leverage \Cref{thm:n_geq_2_main_text} and the differences in circuit structure between different parity sets to find favorable implementations of the parity flow formalism. For example, let us consider the arrangement of physical qubits on a heavy-hex lattice as in \Cref{fig:parity_flow_brickwork}. The qubits marked with both blue and red represent the qubits that represent both parity and base qubits at various times throughout the circuit (those in blue only ever represent base qubits). These base/parity dual qubits are the target of several CNOT gates with each of the neighboring qubits acting as controls. As the corresponding parity sets for this layout satisfy the conditions of \Cref{thm:n_geq_2_main_text}, we know that universal computation can be performed by placing single-qubit $Z$- and $X$-gates at appropriate locations in the circuit, in the manner explained above. The heavy-hex layout considered here is of practical relevance as it is the common layout across a range of current quantum devices made available by IBM Quantum \cite{IBMQ}. Moreover, it is in fact possible to directly map the circuits arising in the parity flow formalism to equivalent circuits using the native gates supported by the IBM Quantum platform, namely by replacing all CNOTs with CZ gates, and by changing all $X$-rotations to $Z$-rotations and vice versa for the dual base/parity qubits.

\subsection{Resources for MBQC} \label{sec:MBQC_res}

\begin{figure*}[!t] 
\centering
\begin{subfigure}[b]{0.48\textwidth} 
\centering
\includegraphics[width=\textwidth]{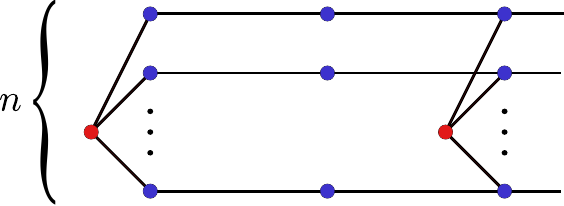} 
\caption{}
\label{fig:universal_graph_even}
\end{subfigure}
\hfill
\begin{subfigure}[b]{0.48\textwidth} 
\centering
\includegraphics[width=\textwidth]{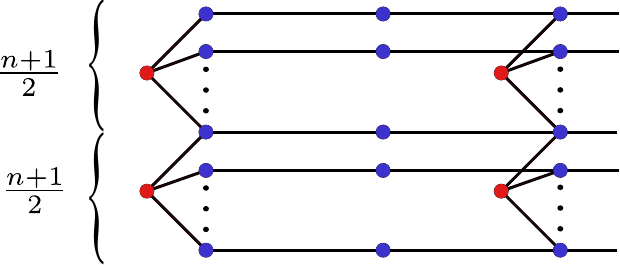} 
\caption{}
\label{fig:universal_graph_odd}
\end{subfigure} 
\caption{Since parity quantum computing can be viewed as a specific type of measurement-based quantum computation (MBQC), the sufficient conditions from \Cref{thm:n_geq_2_main_text} in the main text allow us to derive different families of universal resources for MBQC. In this figure, graph states are depicted relating to the minimal possible generating sets for parity quantum computing, treating the two cases of an even number of base qubits (figure (a)) and an odd number of qubits (figure (b)) separately. In each case, vertices correspond to qubits prepared in the $\ket{+}$ state and edges correspond to CZ gates between them. In figure (a), $n$ different linear cluster states (with blue vertices) are connected via the red vertices which correspond to the parity qubits under the mapping from parity quantum computing to MBQC. In figure (b), the $n$ different linear clusters states are connected via two lots of red vertices each of which connect to $\frac{n+1}{2}$ linear cluster states, with one in common. In both cases, universal MBQC is performed by measuring each red vertex in the $\textrm{YZ}$-plane of the Bloch-sphere and all blue vertices in the $\textrm{XY}$-plane of the Bloch sphere. The measurement order is depicted by the left-to-right ordering of vertices.}
\label{fig:universal_graphs}
\end{figure*}

As mentioned in \Cref{sec:parity}, one method to perform the decoding phase of a parity computation via single-qubit Pauli X-measurements on each of the parity qubits, followed by conditional corrections based on whether a positive (no correction required) or negative (correction required) outcome is obtained \cite{messinger_23}. In Ref.~\cite{smith_24}, it was demonstrated that, in such a case, parity quantum computation corresponds to measurement-based quantum computation (MBQC) using specific classes of bipartite graph states and measurements in the $\textrm{YZ}$-plane of the Bloch sphere. Accordingly, a further consequence of the results presented here is that we can define novel classes of universal resources for MBQC to complement those already known in the literature (see e.g., \cite{Nest_06}), which we elaborate upon below.

Formally, a universal resource for MBQC is a family of states, which we denote by $\Psi$, such that for any state $\ket{\gamma}$ on $n$ qubits there exists a state $\ket{\varphi} \in \Psi$ on $m \geq n$ qubits such that $\ket{\gamma}$ can be obtained deterministically from $\ket{\varphi}$ via local operations and classical communication \cite{Nest_06}. In this context, the local operations and classical communication (LOCC) are taken to be the single-qubit projective measurements and classical feed-forwarding of measurement results upon which MBQC is based (a brief introduction to MBQC and the requirements for deterministic computation in that framework is given in \Cref{app:mbqc}).

So, since our aim is to construct a universal resource for MBQC related to the generating sets presented above, we also need to consider different generating sets $\mathcal{G}_{\text{parity}}$ for different numbers of qubits $n$. As \Cref{thm:n_geq_2_main_text} identifies various different choices of $\mathbb{P}$ and hence different $\mathcal{G}_{\text{parity}}$, there are several different ways one can construct such universal resources; to simplify the discourse, we focus on the choices for $\mathbb{P}$ containing the minimal number of elements. Even with this restriction, there are typically various choices for how to choose $\mathbb{P}$ when $n$ is odd (i.e. for $n > 3$) - in such a case, we make the further choice to take $\mathbb{P}$ to consist of $S_{1}$ and $S_{2}$ such that $|S_{1}| = |S_{2}| = \frac{n+1}{2}$. Let us then define the following notation that makes the $n$ dependence explicit: let
\begin{align}
\mathbb{P}_{*}^{(n)} := \begin{cases} \big\{\{1,2,...,n\}\big\}, &\text{ if } n \bmod 2 = 0 \\
\big\{\{1,2,...,\frac{n+1}{2}\}, \{\frac{n+1}{2},...,n\} \big\}, &\text{ if } n \bmod 2 = 1
\end{cases}
\end{align}
and then let us define $\mathcal{G}_{\text{parity}*}^{(n)} := \mathcal{G}_{\text{s.q.}}^{(n)} \cup \mathcal{G}_{\mathbb{P}_{*}^{(n)}}$, where $\mathcal{G}_{\text{s.q.}}^{(n)} = \{iX_{1}, iZ_{1},...,iX_{n},iZ_{n}\}$ is the same as in \Cref{eq:gen_set_sq} but with the extra superscript allowing reference to different numbers of qubits when necessary. Note also that $\mathbb{P}_{*}^{(n)}$ is the same as that considered in \Cref{subsec:const_dep}, with $j = \frac{n+1}{2}$.

From the results of \Cref{thm:n_geq_2_main_text}, we know that $\mathcal{G}_{\text{parity}}^{(n)}$ is a universal generating set for each $n$. Accordingly, for any given $n$-qubit state $\ket{\gamma}$, there exists a sequence of rotations of elements of $\mathcal{G}_{\text{parity}}^{(n)}$ that produces a unitary $U_{\gamma}$ taking $\ket{+}^{\otimes n}$ to $\ket{\gamma}$. To construct a universal resource for MBQC, it suffices to produce a family of graph states that can implement $U_{\gamma}$ for any $n$ and any $\ket{\gamma}$. 

In fact, we can proceed by combining the graph states that support computations corresponding to $\mathcal{G}_{\text{s.q.}}^{(n)}$ with those that support computations corresponding to $\mathcal{G}_{\mathbb{P}_{*}^{(n)}}$. Explicitly, since the rotations of elements of $\mathcal{G}_{\text{s.q.}}^{(n)}$ are tensor products of single-qubit rotations, the corresponding graph states are tensor products of linear cluster states which are known to support single-qubit rotations via measurements in the $\textrm{XY}$-plane of the Bloch sphere (see e.g., Fig $2$ of \cite{Raussendorf_03}). In particular, to implement a $Z$-rotation followed by an $X$-rotation, it suffices to use a three-qubit cluster state:
\begin{align}
\bra{+_{\phi_{1}}\!+_{\tau_{2}}}CZ_{1,2}CZ_{2,3} \ket{+\!+\!+}_{123} = e^{-i\tau_{2}X}e^{-i\phi_{1}Z}\ket{+}_{3}
\end{align}
where $\bra{+_{\phi_{1}}}$ (resp. $\bra{+_{\tau_{2}}}$) denotes the positive projector for the measurement of the operator $e^{-i\phi_{1} Z}Xe^{i\phi_{1} Z}$ applied to the first qubit (resp. $e^{-i\tau_{2} Z}Xe^{i\tau_{2} Z}$ applied to the second qubit) and where the output of the computation is given as the state of qubit $3$. The tensor products of linear clusters states supporting sequences of rotations of elements from $\mathcal{G}_{\text{s.q.}}^{(n)}$ are depicted in \Cref{fig:universal_graphs} (the sub-graphs containing blue vertices only).

As identified in Ref.~\cite{smith_24}, the graph states supporting rotations of elements of $\mathcal{G}_{\mathbb{P}}$ are bipartite graph states with the qubits related to parity sets forming one partition and the $n$ qubit associated to the logical state forming the other. For the choice $\mathbb{P}_{*}^{(n)}$ that we consider here, that means that the corresponding graph states consist of $n+1$ qubits if $n$ is even and $n+2$ if $n$ is odd. Rotations of elements of $\mathcal{G}_{\mathbb{P}_{*}^{(n)}}$ are then implemented by performing measurements in the $\textrm{YZ}$-plane of the Bloch sphere on the corresponding parity qubits; for example, for $n$ even (i.e. $\mathbb{P}_{*}^{(n)} = \{S_{1}\}$) we have that 
\begin{align}
\bra{0_{\theta_{S_{1}}}} \prod_{j \in S_{1}} CZ_{S_{1},j} &\ket{+}_{S_{1}}\ket{+\dots+}_{1...n} \nonumber \\
&= e^{-i\theta_{S_{1}} Z_{1} \otimes \dots \otimes Z_{n}}\ket{+\dots+}_{1...n}
\end{align}
where $\bra{0_{\theta_{S_{1}}}}$ is the positive projector associated with the operator $e^{-i\theta_{S_{1}} X}Ze^{i\theta_{S_{1}} X}$ applied to the parity qubit and where the computational output is a state on the $n$-qubits making up the partition corresponding to the base (i.e. non-parity) qubits. These bipartite graph states are depicted as sub-graphs of the graphs in \Cref{fig:universal_graphs} comprising the red vertices and their nearest neighbors.

The universal resource is then constructed by combining the linear cluster states and bipartite graph states together in an appropriate way. Specifically, consider $n$ linear clusters states of length $2l$ arranged as in \Cref{fig:universal_graphs}. Then, we identify the $n$ qubits in every second column with the $n$ qubits that form the non-parity partition of a copy of the bipartite graph states above, producing the graph states shown in \Cref{fig:universal_graph_even} for the case of $n$ even and in \Cref{fig:universal_graph_odd} for the case of $n$ odd. Let us denote the resultant graph states as $\ket{\textrm{G}_{n,l}}$ and let $\Psi_{\text{parity}} := \{\ket{G_{n,l}} | n,l \in \mathbb{N} \}$. 

That the family $\Psi_{\text{parity}}$ is a universal resource follows from the fact that $\mathcal{G}_{\text{parity}*}^{(n)}$ is a universal generating set for all $n$. By measuring each red vertex in the $\textrm{YZ}$-plane and each blue qubit in the $\textrm{XY}$-plane (and performing the appropriate corrections when necessary - see \Cref{app:mbqc}), the final output state is
\begin{align}
\underbrace{\prod_{a=1}^{l} \prod_{b=1}^{n}e^{-i\tau_{a,b}X_{b}}e^{-i\phi_{a,b}Z_{b}} \prod_{S_{c} \in \mathbb{P}_{*}^{(n)}} e^{-i\theta_{a,S_{c}} \bigotimes_{d \in S_{c}} Z_{d}}}_{U_{\bm{\theta},\bm{\phi},\bm{\tau}}} \ket{+}^{\otimes n}
\end{align}
where the angles $\theta$ are given by the measurements $e^{-i\theta X}Ze^{i\theta X}$ on the red vertices, the angles $\tau$ are given by the measurements $e^{-i\tau Z}Xe^{i\tau Z}$ on the blue vertices neighboring the red vertices, and the angles $\phi$ are given by the measurements $e^{-i\phi Z}Xe^{i\phi Z}$ on the non-neighbors of the red vertices. The universality of $\mathcal{G}_{\text{parity}*}^{(n)}$ ensures that, for large enough $l$ and appropriately chosen $\theta, \tau$ and $\phi$, we get $U_{\bm{\theta},\bm{\phi},\bm{\tau}} = U_{\gamma}$ meaning that the output state is indeed $\ket{\gamma}$ as required.

As mentioned earlier, the different choices of $\mathbb{P}$ that satisfy \Cref{thm:n_geq_2_main_text} all lead to universal resources for MBQC via an analogous construction to that presented above. To conclude this section, let us briefly comment on one other choice of $\mathbb{P}$, namely $\mathbb{P}_{\text{pairs}}$ (recall \Cref{eq:P_pairs}), which contains the subset $\{ \{i, i+1\}| i = 1, ..., n-1 \}$ that satisfies \Cref{thm:n_geq_2_main_text}. If we construct graph states corresponding to $\mathcal{G}_{\text{parity}}$ for $\mathbb{P}_{\text{pairs}}$, we obtain graph states that have previously been identified as well-suited for implementing the Quantum Approximate Optimization Algorithm (QAOA) (see Ref.~\cite{farhi2014quantum} for the original QAOA and e.g., Ref.~\cite{Proietti_22} for the QAOA-suitable graph states). This is consistent with the focus on QAOA present in several works using the parity quantum computing framework (see e.g., \cite{lechner_20,weidinger2024performance}).

\section{Discussion} \label{sec:disc}

Much of the focus in quantum computation research in the near-term will center around optimizing implementations of certain algorithms to minimize the total resource requirements while satisfying hardware constraints on connectivity. The generating sets presented here establish both the limits of how far the total qubit count can be minimized for the parity quantum computing framework as well as further avenues for implementing connectivity-aware universal computation. 

To conclude this work, we briefly outline further implications of our results that have not been covered above and point out possible directions for future research. We comment on: (i) the implications for compilation strategies for IBM Quantum architectures, (ii) the implications for noise-mitigation strategies for the parity framework that derive from purification protocols in the context of MBQC, and (iii) the possible connections to other areas of quantum information theory, such as quantum cellular automata.

The implications for compilation for IBM Quantum computers arise from the results presented above for Pauli flow (recall \Cref{fig:parity_flow_brickwork}) in conjunction with the compilation algorithm presented in Ref.~\cite{opt_gen_arxiv}. In the latter work, the algorithm \textsc{PauliCompiler} was developed which, given a generating set $\mathcal{G} \subset i\mathcal{P}_{n}^{*}$ and a target Pauli string $P$ as input, outputs a sequence of elements from $\mathcal{G}$ that produce $P$ via nested commutation. In other words, \textsc{PauliCompiler} specifies the unitary circuit using gates specified by $\mathcal{G}$ whose logical effect is to implement the rotation $e^{-i\theta P}$. As the generating sets considered in this work are also strictly contained in $i\mathcal{P}_{n}^{*}$, they can be used with the \textsc{PauliCompiler} (possibly with some pre-processing). In particular, using the \textsc{PauliCompiler} with the generating sets permitting parity flow computations on the heavy-hex lattice may lead to efficient compilation strategies for IBM Quantum devices which typically use such a heavy-hex layout \cite{IBMQ}.

As discussed in \Cref{sec:MBQC_res} above, there are connections between the parity quantum computing framework and MBQC. One consequence of this is that purification protocols developed for graphs states in the latter context (see e.g., \cite{Dur_03,Aschauer_05}) can be readily adapted to the former, as a method of mitigating the effects of noisy implementations of gates in current devices. The properties of these purification protocols, such as their success probability and maximal reachable fidelity, depend on the properties of the graph states being purified. Consequently, the different possible implementations of universal parity quantum computing corresponding to the different generating sets considered here (and hence also to their corresponding graph states) would exhibit different levels of noise mitigation. However, as demonstrated in Ref.~\cite{opt_gen_arxiv}, different Pauli-string generating sets also exhibit compilation rates for a given unitary, which indicates a possible compilation rate vs. noise mitigation rate trade-off for different parity quantum computations. This trade-off is currently a topic of active research.

Finally, there are connections to other areas of quantum computation and information that are also worth highlighting. As demonstrated in \cite{Hendrik_24}, there are connections between MBQC and Clifford quantum cellular automata (CQCA). In the latter context, a Clifford operation satisfying certain properties (translation invariance and locality-preservation - see Ref.~\cite{Hendrik_24}) is repeatedly applied, interspersed with single-qubit rotations. Since the parity flow framework similarly consists of repeatedly applying a unitary comprised of a sequence of CNOT gates and single-qubit rotations, and moreover since connections exist between the parity framework and MBQC, it would be an interesting endeavor to clarify the similarities and differences between the parity flow and the CQCA pictures.

\begin{acknowledgments} We thank Maxime Cautrès for early discussions regarding a question related to Theorem $1$ and Michael Fellner, Anette Messinger and Christophe Goeller for discussions regarding parity quantum computing. This project was funded in whole or in part by the Austrian Science Fund (FWF) [DK-ALM W$1259$-N$27$, SFB BeyondC F7102, SFB BeyondC F7108-N38, WIT9503323, START grant No. Y1067-N27 and I 6011]. For open access purposes, the authors have applied a CC BY public copyright license to any author-accepted manuscript version arising from this submission. This work was supported by the Austrian Research Promotion Agency (FFG Project No. FO999896208).  This project was funded within the QuantERA II Programme that has received funding from the European Union’s Horizon 2020 research and innovation programme under Grant Agreement No. 101017733. This work was also co-funded by the European Union (ERC, QuantAI, Project No. $101055129$). Views and opinions expressed are however those of the author(s) only and do not necessarily reflect those of the European Union or the European Research Council. Neither the European Union nor the granting authority can be held responsible for them.
\end{acknowledgments}

\bibliography{min_PQC}

\appendix
\onecolumngrid

\section{Mapping Pauli Strings to Symplectic Vectors} \label{app:mapping}

In the main text, several of the results made use of specific properties of the set of Pauli strings and their relation to the Lie algebra $\mathfrak{su}(2^{n})$. In fact, these properties allow us to map any discourse regarding nested commutation of Pauli strings to a different setting, namely that of the symplectic space $\mathbb{F}_{2}^{2n}$. One advantage of doing so is that, in many cases, the questions at hand reduces to a question of linear algebra, which can be simpler to work with. In this appendix, we provide an abridged presentation of this mapping compared to that given in the supplementary material of \cite{opt_gen_arxiv}. The reader is referred to that work for further details.

Recall that $\mathcal{P}_{n} := \{ P_{1} \otimes \dots P_{n} | P_{i} \in \{I,X,Y,Z\}\}$, that $\mathcal{P}_{n}^{*} := \mathcal{P}_{n} \setminus \{I^{\otimes n}\}$ and that $\mathfrak{su}(2^{n}) = \textrm{span}_{\mathbb{R}}\{i\mathcal{P}_{n}^{*}\}$. Recall also that, for any $A,B \in i\mathcal{P}_{n}^{*}$, the commutator $[A,B]$ is either the zero operator $\bm{0}$ if $A$ and $B$ commute, or is proportional, up to some real scalar, to some element of $i\mathcal{P}_{n}^{*}$ otherwise. However, for the purposes of demonstrating universality, the proportionality up to a real scalar can safely be ignored since ultimately we are taking linear combinations over $\mathbb{R}$ in any case. Consequently, all the relevant information regarding $[A,B]$ can be encoded in a binary value: $0$ if they commute and $1$ otherwise. This is a key aspect of the mapping. 

The next observation we need is that the single-qubit Pauli operators satisfy $Y = iXZ$. Accordingly, we can rewrite any element $A \in i\mathcal{P}_{n}^{*}$ as
\begin{align}
A = i^{y_{A}+1} \left(\prod_{j=1}^{n} X_{j}^{a_{j}}\right) \left(\prod_{j=1}^{n} Z_{j}^{a_{n+j}} \right) \label{eq:pre_mapping_form}
\end{align}
where: (i) $y_{A} \in \mathbb{N}_{0}$ is the number of tensor factors of $A$ that are a $Y$, and (ii) $\bm{a} \in \mathbb{F}_{2}^{2n}$ is a binary vector such that $a_{j} = 1$ and $a_{n+j} = 0$ if the $j$th tensor factor of $A$ is a $X$, $a_{j} = 0$ and $a_{n+j} = 1$ if the $j$th tensor factor of $A$ is a $Z$, and $a_{j} = a_{n+j} = 1$  if the $j$th tensor factor of $A$ is a $Y$. This is the mapping that we consider: each element $A \in i\mathcal{P}_{n}^{*}$ is mapped to a binary vector $\bm{a} \in \mathbb{F}_{2}^{2n}$. This is a common mapping used in the context of, e.g., stabilizer quantum mechanics, quantum error correction, and classically simulating Clifford circuits efficiently \cite{nielsen2010quantum,gottesman1997stabilizer,aaronson_04,Anders_06}.

So, we have a mapping from $i\mathcal{P}_{n}^{*}$ to $\mathbb{F}_{2}^{2n}$, but so far the latter simply has the structure of a vector space rather than of a symplectic space as promised above. Let us recall that a symplectic vector space is a vector space $V$ over a field $\mathbb{F}$ equipped with a symplectic bilinear form $\Lambda : V \times V \rightarrow \mathbb{F}$. A symplectic bilinear form is a mapping which is (i) linear in each of its two arguments, (ii) satisfies $\Lambda(\bm{v},\bm{v}) = 0$ for all $\bm{v} \in V$, and (iii) if $\Lambda(\bm{v},\bm{w}) = 0$ for all $\bm{w} \in V$, then $\bm{v} = \bm{0}$. In the present context, the symplectic form encodes the commutation relation between elements of $i\mathcal{P}_{n}^{*}$. For example, for $A,B \in i\mathcal{P}_{n}^{*}$, we can use the notational convention in \Cref{eq:pre_mapping_form} to write
\begin{align}
[A,B] &= i^{y_{A} + y_{B} + 2}\left[\left(\prod_{j=1}^{n} X_{j}^{a_{j}}\right)\!\!\left(\prod_{j=1}^{n}Z_{j}^{a_{n+j}}\right)\!\!\left(\prod_{j=1}^{n} X_{j}^{b_{j}}\right)\!\!\left(\prod_{j=1}^{n}Z_{j}^{b_{n+j}}\right) - \left(\prod_{j=1}^{n} X_{j}^{b_{j}}\right)\!\!\left(\prod_{j=1}^{n}Z_{j}^{b_{n+j}}\right)\!\!\left(\prod_{j=1}^{n} X_{j}^{a_{j}}\right)\!\!\left(\prod_{j=1}^{n}Z_{j}^{a_{n+j}}\right) \right] \\
&= i^{y_{A} + y_{B} + 2}\left[(-1)^{\sum_{j=1}^{n}  a_{n+j}b_{j}} -  (-1)^{\sum_{j=1}^{n}a_{j}b_{n+j}} \right] \left(\prod_{j=1}^{n} X_{j}^{a_{j} + b_{j}}\right)\!\!\left(\prod_{j=1}^{n}Z_{j}^{a_{n+j} + b_{n+j}}\right).
\end{align}
The term in the square brackets in the final line encodes all the information about the commutation relation between $A$ and $B$: if $A,B$ commute, then $\sum_{j = 1}^{n} a_{j}b_{n+j} = \sum_{j=1}^{n}a_{n+j}b_{j} \bmod{2}$, which can be equivalently written as $\sum_{j = 1}^{n} a_{j}b_{n+j} + a_{n+j}b_{j} = 0\bmod 2$, otherwise $\sum_{j = 1}^{n} a_{j}b_{n+j} + a_{n+j}b_{j} = 1\bmod 2$. This expression can be written succinctly using the vectors $\bm{a},\bm{b}$ as 
\begin{align}
\bm{a}^{\top}\begin{bmatrix} \bm{0} & I_{n} \\ I_{n} & \bm{0} \end{bmatrix}\bm{b}
\end{align}
where in $\bm{0}$ in the matrix denotes an $n \times n$ block consisting entirely of $0$s while $I_{n}$ denotes a block containing the $n\times n$ identity matrix. This matrix is precisely how we define the relevant symplectic bilinear form for $\mathbb{F}_{2}^{2n}$: for all $\bm{a}, \bm{b} \in \mathbb{F}_{2}^{2n}$, we define
\begin{align}
\Lambda(\bm{a}, \bm{b}) := \bm{a}^{\top}\begin{bmatrix} \bm{0} & I_{n} \\ I_{n} & \bm{0} \end{bmatrix}\bm{b}.
\end{align}
In the following, we will denote the matrix also by $\Lambda$.

Let us take stock of what we have established so far with the mapping between $i\mathcal{P}_{n}^{*}$ and the symplectic space $\mathbb{F}_{2}^{2n}$. Every element $A \in i\mathcal{P}_{n}^{*}$ is mapped to a unique element of $\mathbb{F}_{2}^{2n}$. Note that no element of $i\mathcal{P}_{n}^{*}$ is mapped to the zero element $\bm{0} \in \mathbb{F}_{2}^{2n}$ (this is convenient for considering commuting operators as we will see in a moment). Moreover, every other element of $\mathbb{F}_{2}^{2n}$ is uniquely associated to an element of $i\mathcal{P}_{n}^{*}$ and vice versa. Consequently, the task of generating all of $i\mathcal{P}_{n}^{*}$ as stated in the main text, reduces to generating all of $\mathbb{F}_{2}^{2n}$ under the mapping. For any $A, B \in i\mathcal{P}_{n}^{*}$ with corresponding vectors $\bm{a}, \bm{b} \in \mathbb{F}_{2}^{2n}$, the operator $[A,B]$ is mapped to $\bm{a}^{\top}\Lambda \bm{b}(\bm{a} + \bm{b}) \in \mathbb{F}_{2}^{2n}$. That is, if $A,B$ commute, they get mapped to $\bm{0} \in \mathbb{F}_{2}^{2n}$ since $\bm{a}^{\top}\Lambda \bm{b} = 0$, otherwise they get mapped to $\bm{a} + \bm{b}$ (note that the vector addition is element-wise and binary). This indicates the utility of this mapping: commutation of operators in the setting of $i\mathcal{P}_{n}^{*}$ corresponds to simple vector addition in $\mathbb{F}_{2}^{2n}$.

To conclude this section, let us establish some notation for later use. Below, we will often refer to a symplectic basis of $\mathbb{F}_{2}^{2n}$. A symplectic basis of $\mathbb{F}_{2}^{2n}$ is a basis $\{\bm{v}_{1}, ..., \bm{v}_{2n}\}$ such that for each $\bm{v}_{i}$ there is precisely one $\bm{v}_{j}$ such that $\bm{v}_{i}^{\top}\Lambda \bm{v}_{j} = 1$. We will write any symplectic basis to which we refer using the suggestive notation $\{\bm{x}_{1}, ..., \bm{x}_{n}, \bm{z}_{1}, ..., \bm{z}_{n} \}$ where
\begin{equation}\label{eq:symplectic_basis_relations}
\begin{gathered} 
\bm{x}_{i}^{\top}\Lambda\bm{z}_{j} = \delta_{ij}, \\ \bm{x}_{i}^{\top}\Lambda \bm{x}_{j} = \bm{z}_{i}^{\top}\Lambda\bm{z}_{j} = 0, \forall i,j.
\end{gathered}
\end{equation}
This notation is suggestive as these constraints mirror the commutation relations of the single-qubit operators in $\mathcal{G}_{\text{s.q.}}$, but in general we do not enforce that the $\bm{x}_{j} \in \mathbb{F}_{2}^{2n}$ are the images of the $iX_{j} \in \mathcal{G}_{\text{s.q.}} \subset i\mathcal{P}_{n}^{*}$ and similarly for the $\bm{z}_{j}$.\footnote{Note that any symplectic basis of $\mathbb{F}_{2}^{2n}$ can be mapped to the set of elements that \textit{are} the images of $\mathcal{G}_{\text{s.q.}}$ by a symplectic transform. This corresponds to a Clifford operation at the level of $i\mathcal{P}_{n}^{*}$.} It will be convenient below to make use of the following notation, where $S \subseteq \{1,...,n\}$:
\begin{equation} \label{eq:super_S}
\begin{aligned}
\bm{z}^{S} &:= \sum_{j \in S} \bm{z}_{j},\\
\bm{x}^{S} &:= \sum_{j \in S} \bm{x}_{j}.
\end{aligned}
\end{equation}
This notation is also suggestive, since, if the $\bm{z}_{j}$ are the images of the $iZ_{j} \in \mathcal{G}_{\text{s.q.}}$ then $\bm{z}^{S}$ is the image of $i\bigotimes_{j \in S} Z_{j} \in \mathcal{G}_{\mathbb{P}}$. If $S = \emptyset$, then we define $\bm{x}^{S} = \bm{z}^{S} = \bm{0} \in \mathbb{F}_{2}^{2n}$.

Let $\{\bm{x}_{1},...,\bm{x}_{n}, \bm{z}_{1},...,\bm{z}_{n} \} \subset \mathbb{F}_{2}^{2n}$ be a fixed symplectic basis. We can write any $\bm{v} \in \mathbb{F}_{2}^{2n}$ uniquely (with respect to this basis) as 
\begin{align}
\bm{v} = \sum_{i = 1}^{n} \alpha_{i}\bm{x}_{i} + \beta_{i}\bm{z}_{i} \label{eq:basis_exp}
\end{align}
where $\alpha_{i},\beta_{i} \in \mathbb{F}_{2}$ for all $i = 1, ..., n$. Let us define
\begin{align}
\mathcal{X}(\bm{v}) &:= \{i \in [n]| \alpha_{i} \neq 0 \}, \\ 
\mathcal{Z}(\bm{v}) &:= \{i \in [n]| \beta_{i} \neq 0 \}, \\
\textrm{wt}(\bm{v}) &:= |\mathcal{X}(\bm{v}) \cup \mathcal{Z}(\bm{v})|.
\end{align}
The latter quantity is the equivalent notion in the symplectic space picture of the usual definition of a weight of a Pauli string (i.e. number of non-identity tensor factors). 

In direct parallel to the adjoint map defined for general Lie algebras (recall \Cref{eq:ad_map_Lie}), let us define, for each $\bm{a} \in \mathbb{F}_{2}^{2n}$, the map $\ad_{\bm{a}}(\cdot)$ via
\begin{align}
\ad_{\bm{a}}(\bm{b}) := \bm{a}^{\top}\Lambda\bm{b}(\bm{a} + \bm{b}). 
\end{align}
When a set $\mathrm{G} \subset \mathbb{F}_{2}^{2n}$ is specified, we define (again, directly paralleling the notation in the main text) the quantity
\begin{align}
\mathrm{G}^{\ad^{(r)}} := \{\ad_{\bm{v}_{1}}\ad_{\bm{v}_{2}}\dots\ad_{\bm{v}_{r}}(\bm{v}_{r+1}) | \bm{v}_{1},...,\bm{v}_{r+1} \in \mathrm{G} \} \subseteq \mathbb{F}_{2}^{2n}
\end{align}
in order to be able to make statements about $\mathrm{G} \bigcup_{r=1}^{\infty} \mathrm{G}^{\ad^{(r)}}$.

As we will be considering sequences of adjoint maps, typically for elements from a choice of symplectic basis of $\mathbb{F}_{2}^{2n}$, it will be convenient to have compact notation for representing both sequence of elements as well as the corresponding compositions of adjoint maps. For $S \subseteq \{1, ..., n\}$ with elements $i_{1}, ..., i_{|S|}$, we define
\begin{align}
\bm{x}_{S} &:= \bm{x}_{i_{1}},\bm{x}_{i_{2}},...,\bm{x}_{i_{|S|}}, \\
\ad_{\bm{x}_{S}}(\cdot)& := \ad_{\bm{x}_{i_{1}}}\ad_{\bm{x}_{i_{2}}}\dots\ad_{\bm{x}_{i_{|S|}}}(\cdot),
\end{align}
and similarly for $\bm{z}_{S}$ and $\ad_{\bm{z}_{S}}(\cdot)$. Note the differences in notation to \Cref{eq:super_S} above: the set $S$ in the superscript indicates a sum of elements while the set $S$ in the subscript indicates a sequence.\footnote{We would like to highlight that the use of sets as sub- and superscripts appear in the notation in distinct ways in this work. For example, using a set $S$ as subscript on a lower case $\bm{z}$, representing a vector in $\mathbb{F}_{2}^{2n}$, is {\em not} the same as a subscript on the capital $Z$, representing a Pauli-$Z$ rotation on a parity qubit.} In the cases where the latter notation is used, the specific ordering of the elements of $S$ will not matter; the important factor is that each element appears in the sequence precisely once. 

These notational conventions provide a compact way of writing certain elements and adjoint sequences, which will be useful below. For example, for any $\bm{v} \in \mathbb{F}_{2}^{2n}$ written as in \Cref{eq:basis_exp} with respect to a given symplectic basis, we can write
\begin{align}
\bm{v} = \bm{x}^{\mathcal{X}(\bm{v})} + \bm{z}^{\mathcal{Z}(\bm{v})}.
\end{align}
Furthermore, if $S \subseteq T \subseteq \{1,...,n\}$, we can write
\begin{align}
\ad_{\bm{x}_{S}}(\bm{z}^{T}) \equiv \ad_{\bm{x}_{i_{1}}}...\ad_{\bm{x}_{i_{|S|}}}(\bm{z}^{T}) = \bm{x}^{S} + \bm{z}^{T}.
\end{align}

\section{Proof of \Cref{thm:n_geq_2_main_text}} \label{app:main_thm}

In this section we present the proof of \Cref{thm:n_geq_2_main_text} in the main text. We make use of the mapping and notation presented in \Cref{app:mapping}. In the statement of \Cref{thm:n_geq_2_main_text}, we start with the set $\mathcal{G}_{\text{s.q.}}$ and consider what sets $\mathcal{G}_{\mathbb{P}}$ ensure that $\mathcal{G}_{\text{parity}} := \mathcal{G}_{\text{s.q.}} \cup \mathcal{G}_{\mathbb{P}}$ are universal. Under the mapping outlined above, $i\mathcal{P}_{n}^{*}$ becomes $\mathbb{F}_{2}^{2n} \setminus \bm{0}$, $\mathcal{G}_{\text{s.q.}}$ becomes a symplectic basis of $\mathbb{F}_{2}^{2n}$, and $\mathcal{G}_{\mathbb{P}}$ becomes $\{\bm{z}^{S} | S \in \mathbb{P} \}$. With these changes, \Cref{thm:n_geq_2_main_text} can be stated equivalently as:

\begin{Theorem} \label{thm:n_geq_2} Let $n \in \mathbb{N}_{\geq 2}$ and let $\mathrm{G}' \subset \mathbb{F}_{2}^{2n}$ be a symplectic basis with elements $\{\bm{x}_{1}, ..., \bm{x}_{n}, \bm{z}_{1}, ..., \bm{z}_{n}\}$ which satisfy \Cref{eq:symplectic_basis_relations}. Let $S_{1}, ..., S_{k}$ for some $1 \leq k \leq n-1$ be such that 
\begin{enumerate}
    \item $|S_{i}|$ is even for all $i \in \{ 1, ..., k\}$,
    \item $\bigcup_{i=1}^{k} S_{i} = [n]$, 
    \item if $k \geq 2$, then 
    \begin{enumerate}
        \item $S_{i} \cap S_{j} = \emptyset$ for all $1 \leq j \leq k$ such that $j \neq i, i+1$, and
        \item $S_{i} \cap S_{i+1} = \{s_{i}\}$ for all $i \leq k-1$, with the $s_{i} \in [n]$ all distinct.
    \end{enumerate}
\end{enumerate} 
Then, for $\mathrm{G} := \mathrm{G}' \cup \{\bm{z}^{S_{i}}| i =1, ..., k\}$, we have that
\begin{align}
\mathrm{G} \bigcup_{r=1}^{\infty} \mathrm{G}^{\ad^{(r)}} = \mathbb{F}_{2}^{2n}.
\end{align}
\end{Theorem}

Before proving the theorem, we state and prove two lemmas which are used in the proof of the theorem. The first allows us to reduce the task of showing that $\mathrm{G} \bigcup_{r=1}^{\infty} \mathrm{G}^{\ad^{(r)}} = \mathbb{F}_{2}^{2n}$ even further by making use of the properties of the symplectic basis $\mathrm{G}'$, while the second abstracts some of the structure present in the proof of the theorem.

\begin{Lemma} \label{lem:each_subset} Let $\mathrm{G}' = \{\bm{x}_{1}, ..., \bm{x}_{n}, \bm{z}_{1}, ..., \bm{z}_{n}\} \subseteq \mathbb{F}_{2}^{2n}$ be a symplectic basis satisfying \Cref{eq:symplectic_basis_relations} and let $T \subseteq \{1,...,n\}$. Then for any $\bm{w} \in \mathbb{F}_{2}^{2n}$ that can be written as
\begin{align}
\bm{w} = \sum_{i \in T} \alpha_{i} \bm{x}_{i} + \beta_{i} \bm{z}_{i}
\end{align}
for $\alpha_{i}, \beta_{i} \in \{0,1\}$ are not both zero for each $i \in T$,there exists a sequence of elements $\bm{u}_{1}, ..., \bm{u}_{r} \in \mathrm{G}'$ for some $r$ such that
\begin{align}
\ad_{\bm{u}_{1}}\dots\ad_{\bm{u}_{r}}(\bm{z}^{T}) = \bm{w}.
\end{align}
\end{Lemma}

\begin{proof} By assumption, $\bm{w}$ is such that $\mathcal{X}(\bm{w}), \mathcal{Z}(\bm{w}) \subseteq T$ and moreover that $T \setminus \mathcal{Z}(\bm{w}) \subseteq \mathcal{X}(\bm{w})$ (the latter constraint comes from the fact that $\alpha_{i}$ and $\beta_{i}$ cannot both be $0$). Let $\bm{u}_{1}, ..., \bm{u}_{r}$ be the sequence $\bm{z}_{T \setminus \mathcal{Z}(\bm{w})},\bm{x}_{\mathcal{X}(\bm{w})}$. Using the \Cref{eq:symplectic_basis_relations}, we then have
\begin{align}
\ad_{\bm{u}_{1}}\dots\ad_{\bm{u}_{r}}(\bm{z}^{T}) &= \ad_{\bm{z}_{T \setminus \mathcal{Z}(\bm{w})}}\ad_{\bm{x}_{\mathcal{X}(\bm{w})}}(\bm{z}^{T}) \\
&= \ad_{\bm{z}_{T \setminus \mathcal{Z}(\bm{w})}}(\bm{x}^{\mathcal{X}(\bm{w})} + \bm{z}^{T}) \\
&= \bm{x}^{\mathcal{X}(\bm{w})} + \bm{z}^{T} + \bm{z}^{T \setminus \mathcal{Z}(\bm{w})} \\
&= \bm{x}^{\mathcal{X}(\bm{w})} + \bm{z}^{\mathcal{Z}(\bm{w})} \\
&= \bm{w}
\end{align}
as required.
\end{proof}

\begin{Lemma} \label{lem:z_A_improved} Let $A,B,C \subseteq [n]$ be such that $A,B \neq \emptyset$, $A \subseteq B$, $|B| = 0 \bmod 2$, and $|B \cap C| = 1$ if $C \neq \emptyset$ ($C$ is allowed to be empty, while $A$ and $B$ are not). Let $\overline{A} := B \setminus A$ and, if $A \neq \emptyset$ and $|A| = 0 \bmod 2$, let $a \in A$ be a distinguished element such that $a \notin A \cap C$ and define $\widetilde{A} := \overline{A} \cup \{a\}$. Let us define the sequence
\begin{align}
\bm{p}_{1},...,\bm{p}_{l} := \begin{cases}
\bm{x}_{A},\bm{z}_{A \setminus (B \cap C)},\bm{z}^{B}, \bm{x}_{A\setminus(B \cap C)},\bm{z}^{B},\bm{x}_{B \cap C},\bm{z}^{C} & \text{ if } A \cap C \neq \emptyset, |A| = 1 \bmod 2, \\
\bm{x}_{A},\bm{z}_{A},\bm{z}^{B},\bm{x}_{A\cup (B \cap C)},\bm{z}_{B \cap C},\bm{z}^{B},\bm{x}_{B \cap C},\bm{z}^{C} & \text{ if } A \cap C = \emptyset, |A| = 1 \bmod 2, \\
\bm{x}_{a},\bm{z}^{B},\bm{x}_{\overline{A}},\bm{z}_{\overline{A}},\bm{z}^{B},\bm{x}_{\widetilde{A} \cup (B \cap C)},\bm{z}_{B \cap C},\bm{z}^{B},\bm{x}_{B \cap C},\bm{z}^{C} & \text{ if } A \cap C \neq \emptyset, |A| = 0 \bmod 2, \\
\bm{x}_{a},\bm{z}^{B},\bm{x}_{\overline{A}},\bm{z}_{\overline{A} \setminus (B \cap C)}, \bm{z}^{B}, \bm{x}_{\widetilde{A} \setminus (B \cap C)},\bm{z}^{B}, \bm{x}_{B \cap C},\bm{z}^{C} & \text{ if } A \cap C = \emptyset, |A| = 0 \bmod 2, A \neq \emptyset.
\end{cases} \label{eq:p_sequence}
\end{align}
where any element with a subscript or superscript that is an empty set is simply removed from the sequence. Then
\begin{align}
\ad_{\bm{p}_{1}}...\ad_{\bm{p}_{l-1}}(\bm{p}_{l}) = \bm{z}^{A \cup (C \setminus (B \cap C))}.
\end{align}
\end{Lemma}

\begin{proof} First note that each of the four cases in \Cref{eq:p_sequence} have the same right-hand end of the sequence, namely $\bm{z}^{B},\bm{x}_{B \cap C},\bm{z}^{C}$. If $C$ is empty, this shortens to $\bm{z}^{B}$. When $C$ is non-empty, we will use that, due to the requirement that $|B \cap C| = 1$, $(\bm{z}^{B})^{\top}\Lambda \bm{x}_{B\cap C} = 1$, so we can write
\begin{align}
\ad_{\bm{z}^{B}}\ad_{\bm{x}_{B \cap C}}(\bm{z}^{C}) = \bm{x}^{B \cap C} + \bm{z}^{(B \cup C)\setminus (B \cap C)}.
\end{align}
As the expression above on the right reduces to $\bm{z}^{B}$ for $C = \emptyset$, we can treat both the cases when $C$ is empty and when it is non-empty simultaneously. Furthermore, let us note that, if $|A| = 1 \bmod 2$, then $(\bm{z}^{B})^{\top}\Lambda \bm{x}^{A} = (\bm{z}^{A})^{\top} \Lambda \bm{x}^{A} = 1$, and if $|A| = 0 \bmod 2$, then $(\bm{z}^{B})^{\top} \Lambda \bm{x}^{\widetilde{{A}}} = 1$. Finally, note that if $A \cap C \neq \emptyset$, then $A\cap C = B \cap C$.

For the case where $A \cap C \neq \emptyset$ and $|A| = 1 \bmod 2$, we have that
\begin{align}
\ad_{\bm{p}_{1}}...\ad_{\bm{p}_{l-1}}(\bm{p}_{l}) &= \ad_{\bm{x}_{A}}\ad_{\bm{z}_{A \setminus (B \cap C)}}\ad_{\bm{z}^{B}}\ad_{\bm{x}_{A\setminus(B \cap C)}}(\bm{x}^{B \cap C} + \bm{z}^{(B \cup C)\setminus (B \cap C)}) \\
&= \ad_{\bm{x}_{A}}\ad_{\bm{z}_{A \setminus (B \cap C)}}\ad_{\bm{z}^{B}}(\bm{x}^{A} + \bm{z}^{(B \cup C)\setminus (B \cap C)}) \\
&= \ad_{\bm{x}_{A}}\ad_{\bm{z}_{A \setminus (B \cap C)}}(\bm{x}^{A} + \bm{z}^{C}) \\
&= \ad_{\bm{x}_{A}}(\bm{x}^{A} + \bm{z}^{A \cup C}) \\
&= \bm{z}^{A \cup C} \\
&\equiv \bm{z}^{A \cup (C \setminus (B \cap C))}
\end{align}
where the equivalence in the last line arises from the fact that $A \cap C = B \cap C \neq \emptyset$. For the case where $A \cap C = \emptyset$ and $|A| = 1 \bmod 2$, we have that
\begin{align}
\ad_{\bm{p}_{1}}...\ad_{\bm{p}_{l-1}}(\bm{p}_{l}) &=\ad_{\bm{x}_{A}}\ad_{\bm{z}_{A}}\ad_{\bm{z}^{B}}\ad_{\bm{x}_{A\cup (B \cap C)}}\ad_{\bm{z}_{B \cap C}} (\bm{x}^{B \cap C} + \bm{z}^{(B \cup C)\setminus (B \cap C)}) \\
&=\ad_{\bm{x}_{A}}\ad_{\bm{z}_{A}}\ad_{\bm{z}^{B}}\ad_{\bm{x}_{A\cup (B \cap C)}} (\bm{x}^{B \cap C} + \bm{z}^{B \cup C}) \\
&=\ad_{\bm{x}_{A}}\ad_{\bm{z}_{A}}\ad_{\bm{z}^{B}} (\bm{x}^{A} + \bm{z}^{B \cup C}) \\
&=\ad_{\bm{x}_{A}}\ad_{\bm{z}_{A}} (\bm{x}^{A} + \bm{z}^{C \setminus (B \cap C)}) \\
&=\ad_{\bm{x}_{A}} (\bm{x}^{A} + \bm{z}^{A \cup (C \setminus (B \cap C))}) \\
&=\bm{z}^{A \cup (C \setminus (B \cap C))}.
\end{align}
For the case where $A \cap C \neq \emptyset$ and $|A| = 0 \bmod 2$, we have that
\begin{align}
\ad_{\bm{p}_{1}}...\ad_{\bm{p}_{l-1}}(\bm{p}_{l}) &=\ad_{\bm{x}_{a}}\ad_{\bm{z}^{B}}\ad_{\bm{x}_{\overline{A}}}\ad_{\bm{z}_{\overline{A}}}\ad_{\bm{z}^{B}}\ad_{\bm{x}_{\widetilde{A} \cup (B \cap C)}}\ad_{\bm{z}_{B \cap C}} (\bm{x}^{B \cap C} + \bm{z}^{(B \cup C)\setminus (B \cap C)}) \\
&=\ad_{\bm{x}_{a}}\ad_{\bm{z}^{B}}\ad_{\bm{x}_{\overline{A}}}\ad_{\bm{z}_{\overline{A}}}\ad_{\bm{z}^{B}}\ad_{\bm{x}_{\widetilde{A} \cup (B \cap C)}}(\bm{x}^{B \cap C} + \bm{z}^{B \cup C}) \\
&=\ad_{\bm{x}_{a}}\ad_{\bm{z}^{B}}\ad_{\bm{x}_{\overline{A}}}\ad_{\bm{z}_{\overline{A}}}\ad_{\bm{z}^{B}}(\bm{x}^{\widetilde{A}} + \bm{z}^{B \cup C}) \\
&=\ad_{\bm{x}_{a}}\ad_{\bm{z}^{B}}\ad_{\bm{x}_{\overline{A}}}\ad_{\bm{z}_{\overline{A}}}(\bm{x}^{\widetilde{A}} + \bm{z}^{C \setminus (B \cap C)}) \\
&=\ad_{\bm{x}_{a}}\ad_{\bm{z}^{B}}\ad_{\bm{x}_{\overline{A}}}(\bm{x}^{\widetilde{A}} + \bm{z}^{\overline{A} \cup (C \setminus (B \cap C))}) \\
&=\ad_{\bm{x}_{a}}\ad_{\bm{z}^{B}}(\bm{x}_{a} + \bm{z}^{\overline{A} \cup (C \setminus (B \cap C))}) \\
&=\ad_{\bm{x}_{a}}(\bm{x}_{a} + \bm{z}^{A \cup C}) \\
&=\bm{z}^{A \cup C} \\
&\equiv \bm{z}^{A \cup (C \setminus (B \cap C))}.
\end{align}
For the case where $A \cap C = \emptyset$, $|A| = 0 \bmod 2$ and $A \neq 0$, we have that
\begin{align}
\ad_{\bm{p}_{1}}...\ad_{\bm{p}_{l-1}}(\bm{p}_{l}) &= \ad_{\bm{x}_{a}}\ad_{\bm{z}^{B}}\ad_{\bm{x}_{\overline{A}}}\ad_{\bm{z}_{\overline{A} \setminus (B \cap C)}}\ad_{\bm{z}^{B}}\ad_{\bm{x}_{\widetilde{A} \setminus (B \cap C)}}(\bm{x}^{B \cap C} + \bm{z}^{(B \cup C)\setminus (B \cap C)}) \\
&= \ad_{\bm{x}_{a}}\ad_{\bm{z}^{B}}\ad_{\bm{x}_{\overline{A}}}\ad_{\bm{z}_{\overline{A} \setminus (B \cap C)}}\ad_{\bm{z}^{B}}(\bm{x}^{\widetilde{A}} + \bm{z}^{(B \cup C)\setminus (B \cap C)}) \\
&= \ad_{\bm{x}_{a}}\ad_{\bm{z}^{B}}\ad_{\bm{x}_{\overline{A}}}\ad_{\bm{z}_{\overline{A} \setminus (B \cap C)}}(\bm{x}^{\widetilde{A}} + \bm{z}^{C}) \\
&= \ad_{\bm{x}_{a}}\ad_{\bm{z}^{B}}\ad_{\bm{x}_{\overline{A}}}(\bm{x}^{\widetilde{A}} + \bm{z}^{\overline{A} \cup C}) \\
&= \ad_{\bm{x}_{a}}\ad_{\bm{z}^{B}}(\bm{x}_{a} + \bm{z}^{\overline{A}\cup C}) \\
&= \ad_{\bm{x}_{a}}(\bm{x}_{a} + \bm{z}^{A \cup (C 
\setminus (B \cap C))}) \\
&= \bm{z}^{A \cup (C \setminus (B \cap C))}.
\end{align}
\end{proof}

\begin{proof}[Proof of the theorem] Since it is trivial to produce a sequence that generates $\bm{0} \in \mathbb{F}_{2}^{2n}$, for example by considering the sequence $\bm{x}_{1}, \bm{x}_{1}$, we focus our attention on generating $\mathbb{F}_{2}^{2n} \setminus \{\bm{0}\}$. In fact, we know from \Cref{lem:each_subset}, it suffices to demonstrate that for each subset $T \subseteq \{1,...,n\}$, $\bm{z}^{T} \in \mathrm{G} \bigcup_{r=1}^{\infty} \mathrm{G}^{\ad^{(r)}}$. As the case where $T = \emptyset$ corresponds to $\bm{z}^{T} = \bm{0}$, we only consider $T \neq \emptyset$ henceforth. For $1 \leq i \leq j \leq k$, let us define the following notation:
\begin{align}
S_{i:j} &:= \bigcup_{l = i}^{j} S_{l}, \\
T_{i:j} &:= T \cap S_{i:j}.
\end{align}
If $i = j$, we simply write $S_{j}$ and $T_{j}$.

The proof proceeds by demonstrating the following two points:
\begin{enumerate}[label=(\roman*)]
    \item there exist sequences $\bm{u}_{1}, ...,\bm{u}_{r}$ and $\bm{w}_{1},...,\bm{w}_{t}$ such that
    \begin{align}
    \ad_{\bm{u}_{1}}\dots\ad_{\bm{u}_{r-1}}(\bm{u}_{r}) &= \bm{z}^{T_{1}}, \\
    \ad_{\bm{w}_{1}}\dots\ad_{\bm{w}_{t-1}}(\bm{w}_{t}) &= \bm{z}^{T_{1} \cup \{s_{1}\}}.
    \end{align}
    \item if $k \geq 2$ and there exist sequences $\bm{u'}_{1},...,\bm{u'}_{r'}$ and $\bm{w'}_{1},...,\bm{w'}_{t'}$ such that
    \begin{align}
    \ad_{\bm{u'}_{1}}\dots\ad_{\bm{u'}_{r'-1}}(\bm{u}_{r'}) &= \bm{z}^{T_{1:j}}, \\
    \ad_{\bm{w'}_{1}}\dots\ad_{\bm{w'}_{t'-1}}(\bm{w}_{t'}) &= \bm{z}^{T_{1:j} \cup \{s_{j}\}}
    \end{align}
    for a given $1 \leq j \leq k-1$, then there exists a sequence $\bm{u}_{1},...\bm{u}_{r}$ such that
    \begin{align}
    \ad_{\bm{u}_{1}}\dots\ad_{\bm{u}_{r-1}}(\bm{u}_{r}) &= \bm{z}^{T_{1:j+1}},
    \end{align}
    and, if in addition $j +1 < k$, there also exists a sequence $\bm{w}_{1},...,\bm{w}_{t}$ such that
    \begin{align}
    \ad_{\bm{w}_{1}}\dots\ad_{\bm{w}_{t-1}}(\bm{w}_{t}) &= \bm{z}^{T_{1:j+1} \cup \{s_{j+1}\}}.
    \end{align}
\end{enumerate}

Before proving these two points, let us comment on why they suffice for establishing the theorem. If $k = 1$, then $T_{1} = T$ and we are finished using the first point alone. If $k \ge 2$, then from the first point we know we can produce $\bm{z}^{T_{1}}$ and $\bm{z}^{T_{1} \cup \{s_{1}\}}$, so by the second point we know we can also produce $\bm{z}^{T_{1:2}}$ and $\bm{z}^{T_{1:2} \cup \{s_{2}\}}$. By iteratively applying the second point, we know we can produce $\bm{z}^{T_{1:j}}$ for all $1 \leq j \leq k$, including $j = k$ itself, for which $T_{1:j} = T$. That is, in the case where $k \geq 2$, the proof proceeds inductively, with the first point being the base case and the second point the inductive step.

\textbf{Proof of (i):} If $T_{1} = \emptyset$, defining $\bm{u}_{1},...,\bm{u}_{r}$ to be the sequence $\bm{x}_{1},\bm{x}_{1}$ trivially gives the desired result. Otherwise, we define $\bm{u}_{1},...,\bm{u}_{r}$ to be the sequence $\bm{p}_{1},...\bm{p}_{l}$ from \Cref{lem:z_A_improved} for $A = T_{1}$, $B = S_{1}$ and $C = \emptyset$. As a result of that lemma, we get that 
\begin{align}
\ad_{\bm{u}_{1}}...\ad_{\bm{u}_{r-1}}(\bm{u}_{r}) = \bm{z}^{A \cup (C \setminus (B \cap C))} = \bm{z}^{T_{1}}
\end{align}
as required. 

For the sequence $\bm{w}_{1}, ..., \bm{w}_{t}$, we again use \Cref{lem:z_A_improved}, now with $A = T_{1} \cup \{s_{1}\}$ instead (meaning that $A \neq \emptyset$ regardless of whether or not $T_{1}$ is empty). Defining $\bm{w}_{1},...,\bm{w}_{t}$ to be the resultant sequence $\bm{p}_{1},...\bm{p}_{l}$, we get that 
\begin{align}
\ad_{\bm{w}_{1}}...\ad_{\bm{w}_{t-1}}(\bm{w}_{t}) = \bm{z}^{A \cup (C \setminus (B \cap C))} = \bm{z}^{T_{1} \cup \{s_{1}\}}.
\end{align}

\textbf{Proof of (ii):} Suppose that $k \geq 2$ and that there exist sequences $\bm{u'}_{1}, ..., \bm{u'}_{r'}$ and $\bm{w'}_{1}, ..., \bm{w'}_{t'}$ such that for a given $1 \leq j \leq k-1$, 
\begin{gather}
\ad_{\bm{u'}_{1}}\dots\ad_{\bm{u'}_{r'-1}}(\bm{u'}_{r'}) = \bm{z}^{T_{1:j}}, \\
\ad_{\bm{w'}_{1}}\dots\ad_{\bm{w'}_{t'-1}}(\bm{w'}_{t'}) = \bm{z}^{T_{1:j} \cup \{s_{j}\}}.
\end{gather}
If $T_{1:j} = \emptyset$, then we are in an analogous situation to point (i), and so the same reasoning applies. If $T_{j+1}$ is also empty, then the required sequence is trivial same as above, otherwise we can leverage \Cref{lem:z_A_improved} to obtain the desired sequences by defining $B := S_{j+1}$ and $A := T_{j+1}$ to obtain $\bm{u}_{1},...,\bm{u}_{r}$, and, if in addition $j+1 < k$, by defining $B := S_{j+1}$ and $A := T_{j+1} \cup \{s_{j+1}\}$ to obtain $\bm{w}_{1}, ..., \bm{w}_{t}$. Henceforth, let us assume that $T_{1:j} \neq \emptyset$.

To construct the new sequence $\bm{u}_{1}, ..., \bm{u}_{r}$ and $\bm{w}_{1}, ..., \bm{w}_{t}$ we can make further use of \Cref{lem:z_A_improved}. For the former, we first note that if $T_{j+1} = \emptyset$, then $T_{1:j+1} = T_{1:j}$ so we merely take $\bm{u}_{1},...,\bm{u}_{r}$ to be $\bm{u'}_{1},...,\bm{u'}_{r'}$. Otherwise, taking $A = T_{j+1}$, $B = S_{j+1}$ and $C = T_{1:j} \cup \{s_{j}\}$, we get that $B \cap C = \{s_{j}\}$ and furthermore that the sequence $\bm{p}_{1},...,\bm{p}_{l}$ from \Cref{eq:p_sequence} is such that
\begin{align}
\ad_{\bm{p}_{1}}...\ad_{\bm{p}_{l-1}}(\bm{p}_{l}) = \bm{z}^{A \cup (C \setminus B \cap C)} = \bm{z}^{T_{j+1} \cup (T_{1:j} \setminus \{s_{j}\}) } = \bm{z}^{T_{1:j+1}}.
\end{align}
Noting that in each case in \Cref{eq:p_sequence}, $\bm{p}_{l} = \bm{z}^{T_{1:j} \cup \{ s_{j} \}}$, which is the element produced by the adjoint sequence $\bm{w'}_{1},...,\bm{w'}_{t'}$ by assumption, defining $\bm{u}_{1}, ..., \bm{u}_{r}$ to be the sequence $\bm{p}_{1},...,\bm{p}_{l-1},\bm{w'}_{1},...,\bm{w'}_{t'}$ gives the desired result.

If $j+1 < k$, we can apply similar reasoning to produce the sequence $\bm{w}_{1},..., \bm{w}_{t}$. Taking $A = T_{j+1} \cup \{s_{j+1}\}$, $B = S_{j+1}$ and $C = T_{1:j} \cup \{s_{j}\}$, \Cref{lem:z_A_improved} ensures that the sequence $\bm{p}_{1},...,\bm{p}_{l}$ from \Cref{eq:p_sequence} satisfies
\begin{align}
\ad_{\bm{p}_{1}}...\ad_{\bm{p}_{l-1}}(\bm{p}_{l}) = \bm{z}^{A \cup (C \setminus B \cap C)} = \bm{z}^{(T_{j+1} \cup \{s_{j+1}\}) \cup (T_{1:j} \setminus \{s_{j}\}) } = \bm{z}^{T_{1:j+1} \cup \{s_{j+1}\}}.
\end{align}
As above, taking $\bm{w}_{1}, ..., \bm{w}_{t}$ to be the concatenated sequence $\bm{p}_{1},...,\bm{p}_{l-1},\bm{w'}_{1},...,\bm{w'}_{t'}$ gives the desired result.
\end{proof}

\section{Proof of \Cref{thm:n_odd_no_single_main_text}} \label{app:thm_n_odd_no_single_main_text}

In this section we prove an equivalent form of \Cref{thm:n_odd_no_single_main_text} using the mapping and notation from \Cref{app:mapping}. In particular, we will make use of the quantity $\wt(\cdot)$ which is with defined with respect to a given symplectic basis. We assume that such a basis, denoted $\{\bm{x}_{1},..,\bm{x}_{n},\bm{z}_{1},...,\bm{z}_{n} \}$, is given throughout this section. For the proof of the theorem, we make use of the following lemma:
\begin{Lemma} \label{lem:N_minus_wt_plus_sigma} Let $\bm{u},\bm{v} \in \mathbb{F}_{2}^{2n}$ be such that $\bm{v} \neq \bm{u}$, $\wt(\bm{v}) = n$ and $\ad_{\bm{v}}(\bm{u}) \neq 0$. Then $\wt(\ad_{\bm{v}}(\bm{u})) = n - \wt(\bm{u}) + \sigma$, where $1 \leq \sigma \leq n$ is an odd integer.
\end{Lemma}

\begin{proof} Let us write
\begin{gather}
\bm{v} = \sum_{i=1}^{n} \alpha_{i}\bm{x}_{i} + \beta_{i}\bm{z}_{i}, \\
\bm{u} = \sum_{i \in \mathcal{X}(\bm{u}) \cup \mathcal{Z}(\bm{u})} \gamma_{i} \bm{x}_{i} + \delta_{i}\bm{z}_{i},
\end{gather}
where $\alpha_{i}, \beta_{i}, \gamma_{i},\delta_{i} \in \{0,1\}$. Since $\wt(\bm{v}) = n$, we have that for all $i \in [n]$, $\alpha_{i}$ and $\beta_{i}$ are not both $0$ and similarly, for all $j \in \mathcal{X}(\bm{u}) \cup \mathcal{Z}(\bm{u})$, $\gamma_{j}$ and $\delta_{j}$ are not both $0$. Suppose that $\ad_{\bm{v}}(\bm{u}) \neq \bm{0}$. In particular, this means that
\begin{align}
\bm{v}^{\top}\Lambda \bm{u} = \sum_{i \in \mathcal{X}(\bm{u}) \cup \mathcal{Z}(\bm{u})} \alpha_{i}\delta_{i} + \beta_{i}\gamma_{i} = 1 \bmod 2.
\end{align}
For this to be the case, there must be an odd number of $i \in \mathcal{X}(\bm{u}) \cup \mathcal{Z}(\bm{u})$ such that $\alpha_{i}\delta_{i} + \beta_{i}\gamma_{i} = 1$. This can only occur if either $\alpha_{i}$ and $\delta_{i}$ are both $1$ and at most one of $\beta_{i}$ and $\gamma_{i}$ is $1$, or $\beta_{i}$ and $\gamma_{i}$ are both $1$ and at most one of $\alpha_{i}$ and $\delta_{i}$ is $1$. Let us define $\sigma := |\{i :  \alpha_{i}\delta_{i} + \beta_{i}\gamma_{i} = 1\} |$.

Next, note that $\ad_{\bm{v}}(\bm{u}) \neq \bm{0}$ also means that
\begin{align}
\ad_{\bm{v}}(\bm{u}) = \bm{v} + \bm{u} =\sum_{i \in \mathcal{X}(\bm{u}) \cup \mathcal{Z}(\bm{u})} (\alpha_{i} + \gamma_{i})\bm{x}_{i} + (\beta_{i} + \delta_{i})\bm{z}_{i} + \sum_{j \in [n] \setminus (\mathcal{X}(\bm{u}) \cup \mathcal{Z}(\bm{u}))} \alpha_{j}\bm{x}_{j} + \beta_{j}\bm{z}_{j}.
\end{align}
It follows that 
\begin{align}
\wt(\ad_{\bm{v}}(\bm{u})) = \wt\left(\sum_{i \in \mathcal{X}(\bm{u}) \cup \mathcal{Z}(\bm{u})} (\alpha_{i} + \gamma_{i})\bm{x}_{i} + (\beta_{i} + \delta_{i})\bm{z}_{i} \right) + \wt\left(\sum_{j \in [n] \setminus (\mathcal{X}(\bm{u}) \cup \mathcal{Z}(\bm{u}))} \alpha_{j}\bm{x}_{j} + \beta_{j}\bm{z}_{j} \right).
\end{align}
Due to the requirements on the $\alpha_{i}$ and $\beta_{i}$, we know that 
\begin{align}
\wt\left(\sum_{j \in [n] \setminus (\mathcal{X}(\bm{u}) \cup \mathcal{Z}(\bm{u}))} \alpha_{j}\bm{x}_{j} + \beta_{j}\bm{z}_{j} \right) = |[n]\setminus (\mathcal{X}(\bm{u}) \cup \mathcal{Z}(\bm{u}))| = n - \wt(\bm{u}).
\end{align}
Finally, the implications of the fact that $\bm{v}^{\top}\Lambda \bm{w} = 1$ discussed above ensure that
\begin{align}
\wt\left(\sum_{i \in \mathcal{X}(\bm{u}) \cup \mathcal{Z}(\bm{u})} (\alpha_{i} + \gamma_{i})\bm{x}_{i} + (\beta_{i} + \delta_{i})\bm{z}_{i} \right) = \sigma
\end{align}
which completes the proof.
\end{proof}

The equivalent statement to \Cref{thm:n_odd_no_single_main_text} in the symplectic picture is:
\begin{Theorem} Let $\mathcal{G}$ be a symplectic basis of $\mathbb{F}_{2}^{2n}$ as above. Then for any $\bm{v} \in \mathbb{F}_{2}^{2n}$, we have that
\begin{align}
\langle \mathcal{G} \cup \{\bm{v}\}\rangle_{[\cdot, \cdot]} \subsetneq \mathbb{F}_{2}^{2n}.
\end{align}
\end{Theorem}

\begin{proof} The proof is established by demonstrating that, for any choice of $\bm{v} \in \mathbb{F}_{2}^{2n}$, there is always some element $\bm{w} \in \mathbb{F}_{2}^{2n}$ such that $\bm{w} \notin \langle \mathcal{G} \cup \{\bm{v}\}\rangle_{[\cdot, \cdot]}$. Since the trivial case $\bm{v} = \bm{0}$, results in $\langle \mathcal{G} \cup \{\bm{v}\}\rangle_{[\cdot, \cdot]} = \mathcal{G} \subsetneq \mathbb{F}_{2}^{2n}$, we assume from now on that $\bm{v} \neq \bm{0}$. We consider two cases:
\begin{enumerate}[label=(\roman*)]
    \item $\wt(\bm{v}) < n$;
    \item $\wt(\bm{v}) = n$.
\end{enumerate}
Before presenting the proofs of these cases, let us note that, along with the relations defining the symplectic set  \Cref{eq:symplectic_basis_relations}, we also have that
\begin{align}
\bm{x}_{i}^{\top} \Lambda \bm{v} &\iff i \in \mathcal{Z}(\bm{v}), \\
\bm{z}_{i}^{\top}\Lambda \bm{v} &\iff i \in \mathcal{X}(\bm{v}).
\end{align}

\textbf{Proof of (i):} Suppose that $\bm{v} \in \mathbb{F}_{2}^{2n} \setminus \{\bm{0}\}$ is such that $\wt(\bm{v}) < n$. Then there exists some $j \in [n]$ such that $j \notin \mathcal{X}(\bm{v}) \cup \mathcal{Z}(\bm{v})$. The proof proceeds by showing that
\begin{align}
\bm{w} := \begin{cases} \bm{x}_{i} + \bm{x}_{j}, &\text{ for some } i \in \mathcal{X}(\bm{v}) \text{ if } \mathcal{X}(\bm{v})\neq \emptyset, \\
\bm{z}_{i} + \bm{x}_{j}, &\text{ for some } i \in \mathcal{Z}(\bm{v}) \text{ otherwise.}
\end{cases}
\end{align}
Note the case where both $\mathcal{X}(\bm{v})$ and $\mathcal{Z}(\bm{v})$ are empty is disallowed by the assumption that $\bm{v} \neq \bm{0}$. In the following we assume $\mathcal{X}(\bm{v})$ is non-empty, but the case where only $\mathcal{Z}(\bm{v})$ is non-empty proceeds analogously by replacing $\bm{x}_{i}$ by $\bm{z}_{i}$.

Suppose for a contradiction that $\bm{x}_{i} + \bm{x}_{j} \in \langle \mathcal{G} \cup \{\bm{v}\} \rangle_{[\cdot, \cdot]}$. This means there exists a sequence of elements $\bm{u}_{1}, ..., \bm{u}_{r} \in \mathcal{G} \cup \{\bm{v}\}$ such that
\begin{align}
\ad_{\bm{u}_{1}}...\ad_{\bm{u}_{r-1}}(\bm{u}_{r}) = \bm{x}_{i} + \bm{x}_{j}.
\end{align}
Since $\bm{x}_{i} + \bm{x}_{j} \notin \mathcal{G} \cup \{\bm{v}\}$, it must be that $r \geq 2$. Moreover, since $\bm{x}_{i},\bm{x}_{j}$ are linearly independent, it must be that $\bm{x}_{i} + \bm{x}_{j} \neq \bm{0}$. For $\ad_{\bm{u}_{1}}...\ad_{\bm{u}_{r-1}}(\bm{u}_{r}) \neq \bm{0}$ to hold, it must be that
\begin{align}
\ad_{\bm{u}_{t}}...\ad_{\bm{u}_{r-1}}(\bm{u}_{r}) = \bm{u}_{t} + ... + \bm{u}_{r} \neq \bm{0}
\end{align}
for all $t = 1, ..., r-1$. Furthermore, for $\bm{u}_{1} + ... + \bm{u}_{r}$ to equal $\bm{x}_{i} + \bm{x}_{j}$, it must be the case that there is some $l \in \{1, ..., r\}$ such that $\bm{u}_{l} = \bm{x}_{j}$. However, since for $\bm{u} \in \mathcal{G} \cup\{\bm{v}\}$, $\bm{x}_{j}^{\top}\Lambda \bm{u} = 1$ if and only if $\bm{u} = \bm{z}_{j}$, it follows that the only possible sequences $\bm{u}_{1},...,\bm{u}_{r}$ containing $\bm{x}_{j}$ and satisfying $\ad_{\bm{u}_{t}}...\ad_{\bm{u}_{r-1}}(\bm{u}_{r}) \neq 0$ for all $t = 1,...,r-1$ are of the form
\begin{align}
\bm{u}_{1},...,\bm{u}_{r} = \begin{cases} ...\bm{x}_{j},\bm{x}_{j},\bm{z}_{j},\bm{z}_{j},\bm{x}_{j},\bm{x}_{j},\bm{z}_{j},\bm{x}_{j} \\ 
...\bm{z}_{j},\bm{z}_{j},\bm{x}_{j},\bm{x}_{j},\bm{z}_{j},\bm{z}_{j},\bm{x}_{j},\bm{z}_{j} 
\end{cases}.
\end{align}
That is, the sequence is comprised only of $\bm{x}_{j}$ and $\bm{z}_{j}$, and only in a specific alternating form. In any case, the resultant adjoint sequence is
\begin{align}
\ad_{\bm{u}_{1}}...\ad_{\bm{u}_{r-1}}(\bm{u}_{r}) = \begin{cases} \bm{x}_{j}, \\ \bm{z}_{j}, \\ \bm{x}_{j} + \bm{z}_{j},
\end{cases}
\end{align}
none of which equal $\bm{x}_{i} + \bm{x}_{j}$, providing the desired contradiction.

\textbf{Proof of (ii):} Now suppose that $\bm{v}$ is such that $\wt(\bm{v}) = n$. The proof proceeds by demonstrating the claim that $\wt(\bm{u}_{1}, ..., \bm{u}_{r})$ is odd for any sequence $\bm{u}_{1}, ..., \bm{u}_{r} \in \mathcal{G} \cup \{v\}$ such that $\ad_{\bm{u}_{1}}...\ad_{\bm{u}_{r-1}}(\bm{u}_{r}) \neq \bm{0}$. This proves case (ii) since any $\bm{w} \in \mathbb{F}_{2}^{2n}$ with $\wt(\bm{w})$ even is such that $\bm{w} \notin \langle\mathcal{G} \cup \{\bm{v}\}\rangle_{[\cdot,\cdot]}$.

We prove the claim by induction on the length of the sequence, and by making use of \Cref{lem:N_minus_wt_plus_sigma}. Since $n$ is odd, the elements of $\mathcal{G} \cup \{v\}$ all have odd weight, so the base case of sequences of length $1$ holds true. Suppose that $\wt(\ad_{\bm{u}_{1}}...\ad_{\bm{u}_{r-1}}(\bm{u}_{r}))$ is odd for any sequence $\bm{u}_{1}, ..., \bm{u}_{r} \in \mathcal{G} \cup \{v\}$ such that $\ad_{\bm{u}_{1}}...\ad_{\bm{u}_{r-1}}(\bm{u}_{r}) \neq \bm{0}$ and let $\bm{u} \in \mathcal{G} \cup \{v\}$. Since $\ad_{\bm{u}_{1}}...\ad_{\bm{u}_{r-1}}(\bm{u}_{r}) \neq \bm{0}$, we have that $\ad_{\bm{u}_{1}}...\ad_{\bm{u}_{r-1}}(\bm{u}_{r}) = \bm{u}_{1} + ... + \bm{u}_{r}$, and hence also that $\wt(\bm{u}_{1} + ... + \bm{u}_{r})$ is odd by the assumption. There are two case to consider: (a) $\bm{u} = \bm{x}_{i}$ or $\bm{u} = \bm{z}_{i}$ for some $i \in [n]$ or (b) $\bm{u} = \bm{v}$. Since the claim is concerned with adjoin sequences that are non-zero, let us assume that $\bm{u}^{\top}\Lambda (\bm{u}_{1} + ... + \bm{u}_{r}) = 1$.

For (a), if $\bm{u} = \bm{x}_{i}$, then $\bm{u}^{\top}\Lambda (\bm{u}_{1} + ... + \bm{u}_{r}) = 1$ if and only if $\bm{z}_{i}$ appears in the sequence an odd number of times (including in the expansion of any $\bm{u}_{l} = \bm{v}$ using \Cref{eq:basis_exp}). Similarly, if $\bm{u} = \bm{z}_{i}$, then $\bm{u}^{\top}\Lambda (\bm{u}_{1} + ... + \bm{u}_{r}) = 1$ if and only if $\bm{x}_{i}$ appears in the sequence an odd number of times. In either case, if $\bm{u}^{\top}\Lambda (\bm{u}_{1} + ... + \bm{u}_{r}) = 1$, then $\wt(\bm{u} + \bm{u}_{1} + ... + \bm{u}_{r}) = \wt(\bm{u}_{1} + ... + \bm{u}_{r})$ and hence is odd. 

For (b), since we are assuming that $\bm{u}^{\top}\Lambda(\bm{u}_{1}, ..., \bm{u}_{r}) = 1$, we are in the scenario covered by \Cref{lem:N_minus_wt_plus_sigma}. Accordingly, we know that 
\begin{align}
\wt(\bm{u} + \bm{u}_{1} + ... + \bm{u}_{r}) = n - \wt(\bm{u}_{1} + ... + \bm{u}_{r}) + \sigma
\end{align}
for some odd integer $\sigma$. Since $n$ and $\wt(\bm{u}_{1} + ... + \bm{u}_{r})$ are both odd, the right-hand side of the above expression is guaranteed to be odd. This completes the inductive proof of the claim and thus also the theorem.
\end{proof}

\section{Measurement-Based Quantum Computation} \label{app:mbqc}

In \Cref{sec:MBQC_res} of the main text, we presented a family of graph states related to the generating sets presented in \Cref{thm:n_geq_2_main_text}, and claimed that this family represents a universal resource for MBQC. Recall that a universal resource for MBQC is a family of states $\Psi$, such that for any state $\ket{\gamma}$ on $n$ qubits there exists a state $\ket{\varphi} \in \Psi$ on $m \geq n$ qubits such that $\ket{\gamma}$ can be obtained deterministically from $\ket{\varphi}$ via local operations and classical communication (LOCC). For the family $\Psi = \{\ket{G_{n,l}} | n,l \in \mathbb{N}\}$, the first requirement of the definition of universal resource, namely that any state $\ket{\gamma}$ can be obtained from some $\ket{G_{n,l}}$, follows from \Cref{thm:n_geq_2_main_text} and the construction of $\ket{G_{n,l}}$ as discussed in the main text. Establishing the remaining requirement, that $\ket{\gamma}$ can be obtained {\em deterministically}, is the purpose of this appendix. This necessitates a brief review of some background on MBQC, with which we begin.

In MBQC, a computation is specified by the following: (i) a choice of graph describing the graph state for the computation, (ii) designated subsets of vertices of the graph whose qubits represent the input and output of the computation, (iii) a set of single-qubit measurements performed on the graph state, the positive projectors of which produce the desired logical unitary applied to the input. However, since quantum-mechanical measurements are inherently probabilistic in general, there is no guarantee that the positive outcome will be obtained for each measurement. Accordingly, there is a further requirement for performing the desired unitary with certainty, namely (iv) that there is an ordering of the measurements and a method for conditionally adapting them if a negative measurement outcome occurs, so that the overall result is the same logical unitary operation. This method is based on the properties of the stabilizers of the graph state being considered, which we now briefly review.

Let $\ket{G}$ denote a graph state corresponding to the simple, connected graph $G$ with vertex set $V$ and edge set $E$. For each $v \in V$, the graph state $\ket{G}$ satisfies the equation
\begin{align}
\underbrace{X_{v} \bigotimes_{v' \in N_{v}^{G}} Z_{v'}}_{=: K_{v}}\ket{G} = \ket{G}
\end{align}
where $N_{v}^{G}$ denotes the neighborhood of $v$ in $G$. As a consequence of this equation, it is possible to ``correct'' for the occurrence of a negative measurement outcome, by the following reasoning. Suppose the qubit corresponding to vertex $w$ is measurement in the basis $\{\ket{+_{\theta}}, \ket{-_{\theta}} \}$, i.e., it is measured in the $\textrm{XY}$-plane of the Bloch sphere. Let $v \in V$ be a neighbor of $w$ in $G$. Then, since $\ket{-_{\theta}} = Z\ket{+_{\theta}}$, we have that
\begin{align}
\bra{-_{\theta}}_{w}\ket{G} = \bra{+_{\theta}}_{w}Z_{w}\ket{G} = \bra{+_{\theta}}_{w} X_{v} \bigotimes_{v' \in N_{v}^{G}\setminus w}Z_{v'}\ket{G}.
\end{align}
In this case, we see that we have effectively obtained the desired measurement outcome on $w$ at the cost of needing to apply the remaining operations of $K_{v}$. However, if the measurements on the qubits that would receive these operations are suitably chosen and yet to be performed, these operations can be effected by changing the measurement bases appropriately. For example, suppose that qubit $v$ and each qubit $v' \in N_{v}^{G} \setminus w$ are to be measured in the $\textrm{XY}$-plane with respect to the bases $\{ \ket{+_{\theta_{v}}}, \ket{-_{\theta_{v}}}\}$ and $\{\ket{+_{\theta_{v'}}}, \ket{-_{\theta_{v'}}} \}$ for $v' \in N_{v}^{G}\setminus w$. Using that $X\ket{\pm_{\theta}} = \ket{\pm_{-\theta}}$ and $Z\ket{\pm_{\theta}} = \ket{\pm_{\theta + \pi}}$, we see that 
\begin{align}
\bra{-_{\theta_{w}}}_{w} \bra{+_{\theta_{v}}}_{v} \bigotimes_{v' \in N_{v}^{G}\setminus w}\bra{+_{\theta_{v'}}}_{v'} \ket{G} = \bra{+_{\theta_{w}}}_{w} \bra{+_{(-\theta_{v})}}_{v} \bigotimes_{v' \in N_{v}^{G}\setminus w}\bra{+_{(\theta_{v'}+\pi)}}_{v'} \ket{G}.
\end{align}
This example has assumed measurements in the $\textrm{XY}$-plane of the Bloch sphere, however measurements in the $\textrm{XZ}$- and $\textrm{YZ}$-plane are also possible, by that analogous reasoning (for measurements in the $\textrm{YZ}$-plane, the positive and negative projectors differ by the application of an $X$, while those of measurements in the $\textrm{XZ}$-plane differ by a $Y$; in the latter case, this requires a product of the operators $K_{v}$). In each case, there are common features of the correcting process, namely that (1) the negative outcome is ``corrected'' by ``completing'' a stabilizer $K_{v}$ (or product thereof) on other qubits of the graph state, and (2) that the measurements on those qubits are yet to be performed.

So, to ensure that a computation can proceed with certainty on a given graph state, we need to find a sequence of measurements such that {\em every} measurement can be corrected by the above method. It turns out, that whether this is possible or not is a property of the graph underlying the graph state in conjunction with the assignment of planes of the Bloch sphere (i.e., $\textrm{XY}$, $\textrm{XZ}$ and $\textrm{YZ}$) to each of the qubits indicating which type of measurement is to be performed there. If it is possible to correct for every measurement, then the graph is said to have gflow \cite{browne2007generalized}, which is characterized in the following definition:
\begin{Definition} \label{def:gflow} Let $G = (V, E)$ be a graph, $I$ and $O$ be input and output subsets of $V$ respectively, and $\omega: V \setminus O \rightarrow \{\textrm{XY}, \textrm{XZ}, \textrm{YZ}\}$ be a map assigning measurement planes to qubits. The tuple $(G,I,O, \omega)$ has \textit{gflow} if there exists a map $g: V \setminus O \rightarrow 2^{V \setminus I}$, where $2^{V \setminus I}$ denotes the powerset of $V \setminus I$, and a partial order over $V$ such that the following hold for all $v \in V \setminus O$:
\begin{enumerate}
	\item if $v' \in g(v)$ and $v' \neq v$, then $v < v'$;
	\item if $v' \in \Odd(g(v))$ and $v' \neq v$, then $v < v'$;
	\item if $\omega(v) = \textrm{XY}$, then $v \notin g(v)$ and $v \in \Odd(g(v))$;
	\item if $\omega(v) = \textrm{XZ}$, then $v \in g(v)$ and $v \in \Odd(g(v))$;
	\item if $\omega(v) = \textrm{YZ}$, then $v \in g(v)$ and $v \notin \Odd(g(v))$;
\end{enumerate}
where $\Odd(K) := \{\tilde{v} \in V : |N_{\tilde{v}}^{G} \cap K| = 1 \bmod 2 \}$ for any $K \subseteq V$.
\end{Definition}
In effect, the map $g$ specifies the stabilizer or product of stabilizers corresponding to the correction for each qubit, i.e., the correction occurs by completing $\prod_{v' \in g(v)}K_{v'}$ for each $v$.

It was shown in Ref.~\cite{browne2007generalized} that the presence of gflow is a necessary and sufficient condition for ensuring determinism in MBQC. Accordingly, to finish demonstrating that the family of graph states presented in the main text is a valid universal resource, it suffices to show that each member graph state has gflow. In general, there are many different choices for gflow for a given graph state, which has led to e.g., the development of blind quantum computing protocols \cite{mantri2017flow}, albeit ones with imperfect security guarantees \cite{smith2023min}. For our purposes, we need only specify one choice of gflow for each member graph state, which we turn to now.

\begin{figure*}[!t] 
\centering
\begin{subfigure}[b]{0.8\textwidth} 
\centering
\includegraphics[width=\textwidth]{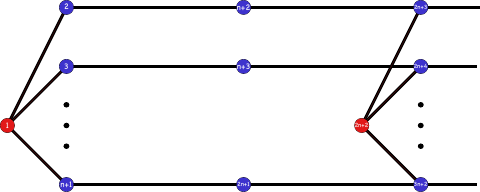} 
\caption{}
\label{fig:universal_graph_even_labeled}
\end{subfigure}
\\
\vspace{0.5cm}
\begin{subfigure}[b]{0.8\textwidth} 
\centering
\includegraphics[width=\textwidth]{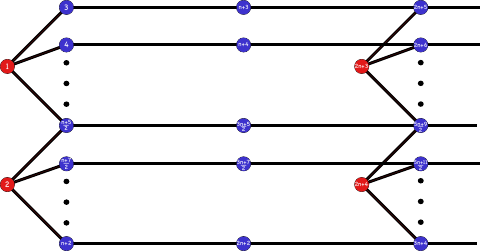} 
\caption{}
\label{fig:universal_graph_odd_labeled}
\end{subfigure} 
\caption{This figure depicts fragments of the graph states $\ket{G_{n,l}}$ that form the universal resource introduced in the main text, but with vertices labeled to aid the definition of gflow presented in this appendix. Figure (a) depicts $\ket{G_{n,l}}$ for $n$ even and (b) depicts $\ket{G_{n,l}}$ for $n$ odd.}
\label{fig:universal_graphs_labeled}
\end{figure*}

Consider \Cref{fig:universal_graphs_labeled}, which depicts the same graph states $\ket{G_{n,l}}$ as \Cref{fig:universal_graphs} in the main text, but now with labels assigned to the vertices. Recall that each vertex colored red is to be measured in the $\textrm{YZ}$-plane of the Bloch sphere while every vertex colored blue is to be measured in the $\textrm{XY}$-plane. We will consider the left-most column of blue vertices as the input set and the right-most column of blue vertices as the output set. We can define a gflow $\ket{G_{n,l}}$ with partial order given by the lexicographical order of the labels and with map $g$ defined by sending each vertex $v$ colored red to the set $\{v\}$ and every vertex $w$ colored blue to the set $\{w'\}$ where $w'$ is the nearest neighbor of $w$ to the right. The remainder of this appendix defines this gflow formally.

We are considering the graph $G_{n,l}$ with vertex set $V$ given by
\begin{align}
V = \begin{cases} \{1,\dots,l(2n+1)\}, &\text{ if } n = 0 \bmod 2, \\
\{1,\dots,l(2n+2) \}, &\text{ if } n = 1\bmod2.
\end{cases}
\end{align}
We can specify the edge set $E$ by stating the sets of neighbors for each vertex, which is more convenient for verifying the choice of partial order and mappings defined below are valid gflows. For $n$ even, we have that
\begin{align}
N_{v}^{G} = \begin{cases} \{v+1,\dots,v+n \}, &\text{ if } v = j(2n+1) + 1 \text{ for some } j \in \{0,\dots,l-1\}, \\
\{1, v+n\}, &\text{ if } v= i \text{ for } i\in\{2,\dots,n+1\},\\
\{v-n-1,j(2n+1)+1, v+n\}, &\text{ if } v= j(2n+1) + i \text{ for } j \in \{1,\dots,l-1\}, i\in\{2,\dots,n+1\},\\
\{v-n\}, &\text{ if } v = (l-1)(2n+1)+i \text{ for } i \in \{n+2,\dots,2n+1\}, \\
\{v-n,v+n+1\}, &\text{ if } v = j(2n+1)+i \text{ for } j \in \{0,\dots,l-2\}, i \in \{n+2,\dots,2n+1\}. 
\end{cases} \label{eq:even_neighborhood}
\end{align}
For $n$ odd, we have that
\begin{align}
N_{v}^{G} = \begin{cases} \{v+2,\dots,v+1+\frac{n+1}{2}\}, &\text{ if } v = j(2n+2) + 1 \text{ for } j \in \{0,\dots,l-1\}, \\
\{v+1+\frac{n+1}{2},\dots,v+n\}, &\text{ if } v = j(2n+2) + 2 \text{ for } j \in \{0,\dots,l-1\}, \\
\{1,v+n\}, &\text{ if } v= i \text{ for } i\in\{3,\dots,1+\frac{n+1}{2}\},\\
\{2,v+n\}, &\text{ if } v= i \text{ for }  i\in\{3+\frac{n+1}{2},\dots,n+2\},\\
\{1,2,v+n\}, &\text{ if } v= 2 + \frac{n+1}{2},\\
\{v-n-2,j(2n+2)+1,v+n\}, &\text{ if } v= j(2n+2) + i \text{ for } j \in \{1,\dots,l-1\}, i\in\{3,\dots,1+\frac{n+1}{2}\},\\
\{v-n-2,j(2n+2)+2,v+n\}, &\text{ if } v= j(2n+2) + i \text{ for } j \in \{1,\dots,l-1\}, i\in\{3+\frac{n+1}{2},\dots,n+2\},\\
\{v-n-2,j(2n+2)+1,j(2n+2)+2,v+n\}, &\text{ if } v= j(2n+2) + 2 + \frac{n+1}{2}\text{ for } j \in \{1,\dots,l-1\},\\
\{v-n\}, &\text{ if } v = (l-1)(2n+2)+i \text{ for } i \in \{n+3,\dots,2n+2\},\\
\{v-n, v+n+2\}, &\text{ if } v = j(2n+2)+i \text{ for } j \in \{0,\dots,l-2\}, i \in \{n+3,\dots,2n+2\}.
\end{cases} \label{eq:odd_neighborhood}
\end{align}
The input and output sets are defined to be
\begin{align}
I = \begin{cases} \{2,\dots,n+1\}, &\text{ if } n = 0 \bmod 2, \\
\{3,\dots,n+2 \}, &\text{ if } n = 1 \bmod 2,
\end{cases}
\end{align}
and 
\begin{align}
O = \begin{cases}\{(l-1)(2n+1) + n + 2,\dots,l(2n+1)\}, &\text{ if } n = 0 \bmod 2, \\
\{(l-1)(2n+2) + n + 3,\dots,l(2n+2) \}, &\text{ if } n = 1 \bmod 2.
\end{cases}
\end{align}
In the case where $n$ is even, the map $\omega$ that specifies the planes of measurement is defined by $\omega(v) = \textrm{YZ}$ if $v = j(2n+1) + 1$ for some $j \in \{0,...,l-1\}$, with all other vertices being assigned $\textrm{XY}$. Similarly, in the case where $n$ is odd, $\omega$ is defined by $\omega(v) = \textrm{YZ}$ if $v = j(2n+2) + 1$ or $v = j(2n+2) + 2$ for some $j \in \{0,...,l-1\}$, and $\omega(v) = \textrm{XY}$ otherwise.

In both cases, the partial order on vertices can be taken to be the order of the corresponding labels, i.e. $1 < 2 < \dots < l(2n+1)$ in the even $n$ case and $1 < 2 < \dots < l(2n+2)$ in the odd case. We define the map $g_{\text{even}}$ via
\begin{align}
g_{\text{even}}(v) = \begin{cases} \{v\}, &\text{ if } v = j(2n+1) + 1 \text{ for some } j \in \{0,\dots,l-1\}, \\
\{v + n\}, &\text{ if } v= j(2n+1) + i \text{ for } j \in \{0,\dots,l-1\}, i\in\{2,\dots,n+1\},\\
\{v + n+1\}, &\text{ if } v = j(2n+1)+i \text{ for } j \in \{0,\dots,l-1\}, i \in \{n+2,\dots,2n+1\}.
\end{cases}
\end{align}
Similarly, the map $g_{\text{odd}}$ is defined via
\begin{align}
g_{\text{even}}(v) = \begin{cases} \{v\}, &\text{ if } v = j(2n+2) + i \text{ for } j \in \{0,\dots,l-1\},i\in\{1,2\}, \\
\{v + n\}, &\text{ if } v= j(2n+2) + i \text{ for } j \in \{0,\dots,l-1\}, i\in\{3,\dots,n+2\},\\
\{v + n+1\}, &\text{ if } v = j(2n+2)+i \text{ for } j \in \{0,\dots,l-1\}, i \in \{n+3,\dots,2n+2\}.
\end{cases}
\end{align}
Since both $g_{\text{even}}$ and $g_{\text{odd}}$ assign singleton sets to each vertex, it follows that for each $v$, $\Odd(g_{\text{even}}(v))$ (resp. $\Odd(g_{\text{odd}}(v))$) is just the sets of neighbors of the vertex in $g_{\text{even}}(v)$ (resp. $g_{\text{odd}}(v)$). By observing \Cref{eq:even_neighborhood} and \Cref{eq:odd_neighborhood}, it can be verified that both $g_{\text{even}}$ and $g_{\text{odd}}$ satisfy the requirements of \Cref{def:gflow} for the choice of partial order outlined above. For example, for $n$ even, we see that the elements of $g_{\text{even}}(v) \setminus v$ and $N_{g_{\text{even}}(v)}^{G} \setminus v$ all have labels greater than $v$ and hence are later than $v$ in the choice of partial order, meaning that conditions $1$ and $2$ are satisfied. For $v = j(2n+1) + 1$ with $j \in \{0, ..., l-1\}$, we have that $\omega(v) = \textrm{YZ}$ and moreover that $g_{\text{even}}(v) = \{v\}$ and that $\Odd(g_{\text{even}}(v)) = \{v+1,\dots,v+n\}$, so condition $5$ holds. For all other $v$, we have that $\omega(v) = \textrm{XY}$ and that $g_{\text{even}}(v)$ contains a neighbor of $v$ meaning that $v \notin g_{\text{even}}(v)$ but $v \in \Odd(g_{\text{even}}(v))$ so condition $3$ is satisfied. The reasoning for $g_{\text{odd}}$ proceeds analogously.

\end{document}